\documentclass[a4paper]{article}

\usepackage{acronym}
\usepackage{amsmath}
\usepackage{amssymb}
\usepackage{amsthm}
\usepackage{algorithm}
\usepackage{algorithmicx}
\usepackage{authblk}
\usepackage[american]{babel}
\usepackage[format=plain,position=bottom,labelfont=bf,textfont=small]{caption}
\usepackage{circuitikz}
\usepackage{dsfont}
\usepackage[inline]{enumitem}
\usepackage[top=1in,left=1in,bottom=1in,right=1in]{geometry}
\usepackage{graphicx}
\usepackage{multirow}
\usepackage[small]{subfigure}
\usepackage{tikz}
\ifdefined \arXiv
	\usepackage[disable]{todonotes}
\else
	\usepackage{todonotes}
\fi
\usepackage{xspace}
\usepackage[pdfborder={0 0 0}]{hyperref} % last \usepackage

\acrodef{ADC}{Analog-to-Digital Converter}
	\acrodefindefinite{ADC}{an}{an}
\acrodef{ASIC}{Application-Specific Integrated Circuit}
	\acrodefindefinite{ASIC}{an}{an}
\acrodef{BRGC}{Binary Reflected Gray Code}
\acrodef{CMOS}{Complementary Metal-Oxide-Semiconductor}
\acrodef{CMUX}{Metastability-Containing Multiplexer}
\acrodef{DAG}{Directed Acyclic Graph}
\acrodef{DCO}{Digitally Controlled Oscillator}
\acrodef{FPGA}{Field-Programmable Gate Array}
	\acrodefindefinite{FPGA}{an}{a}
\acrodef{MTBF}{Mean Time Between Failures}
	\acrodefindefinite{MTBF}{an}{a}
\acrodef{MUX}{Multiplexer}
\acrodef{KV}{Karnaugh-Veitch}
\acrodef{SEU}{Single Event Upset}
	\acrodefindefinite{SEU}{an}{a}
\acrodef{TC}{Thermometer Code}
\acrodef{TDC}{Time-to-Digital Converter}
\acrodef{TTP}{Time-Triggered Protocol}
\acrodef{VLSI}{Very Large-Scale Integration}

\newtheorem{theorem}{Theorem}
\newtheorem{lemma}[theorem]{Lemma}
\newtheorem{corollary}[theorem]{Corollary}
\newtheorem{definition}[theorem]{Definition}

\newtheorem{observation}[theorem]{Observation}
\newtheorem{example}[theorem]{Example}

\setlist[enumerate,1]{label=(\arabic*)}
\setlist[enumerate,2]{label=(\alph*),ref=(\arabic{enumi}\alph*)}

\graphicspath{{./graphics/}}

\usetikzlibrary{circuits,calc}

\newcommand{\B}{\mathds{B}}
\newcommand{\BM}{\B_\meta}
\def\dash---{\kern.16667em---\penalty\exhyphenpenalty\hskip.16667em\relax}

\newcommand{\Eval}{\operatorname{Eval}}
\newcommand{\Fun}{\operatorname{Fun}}
\newcommand{\gand}{\ensuremath{\operatorname{\textsc{And}}}\xspace}

\newcommand{\gnand}{\ensuremath{\operatorname{\textsc{Nand}}}\xspace}
\newcommand{\gnot}{\ensuremath{\operatorname{\textsc{Not}}}\xspace}
\newcommand{\gor}{\ensuremath{\operatorname{\textsc{Or}}}\xspace}
\newcommand{\gone}{\ensuremath{\operatorname{\textsc{Const1}}}\xspace}
\newcommand{\gxor}{\ensuremath{\operatorname{\textsc{Xor}}}\xspace}
\newcommand{\gzero}{\ensuremath{\operatorname{\textsc{Const0}}}\xspace}
\newcommand{\In}{\operatorname{In}}
\newcommand{\Loc}{\operatorname{Loc}}
\newcommand{\meta}{\ensuremath{\textsc{M}}\xspace}
\newcommand{\N}{\mathds{N}}
\newcommand{\nop}[1]{}
\newcommand{\Out}{\operatorname{Out}}
\newcommand{\Pow}{\operatorname{\mathcal{P}}}

\newcommand{\Read}{\operatorname{Read}}
\newcommand{\Res}{\operatorname{Res}}
\newcommand{\ResM}{\Res_\meta}
\newcommand{\rg}{\operatorname{rg}}
\newcommand{\rgtoun}{\operatorname{rg2un}}
\newcommand{\un}{\operatorname{un}}
\newcommand{\Write}{\operatorname{Write}}

\title{Metastability-Containing Circuits}
\date{}

\author[1,2]{Stephan~Friedrichs}
\author[3]{Matthias~F\"ugger}
\author[1]{Christoph~Lenzen}
\affil[1]{Max Planck Institute for Informatics, Saarland Informatics Campus, Germany\newline
	Email:~\texttt{\{sfriedri,clenzen\}@mpi-inf.mpg.de}}
\affil[2]{Saarbr\"ucken Graduate School of Computer Science}
\affil[3]{CNRS, LSV, ENS Paris-Saclay, Email:~\texttt{mfuegger@lsv.fr}}

\begin{document}
\maketitle

\begin{abstract}
	In digital circuits, \emph{metastability} can cause deteriorated signals that neither are logical 0 or logical~1, breaking the abstraction of Boolean logic.
	Unfortunately, any way of reading a signal from an unsynchronized clock domain or performing an analog-to-digital conversion incurs the risk of a metastable upset;
	no digital circuit can deterministically avoid, resolve, or detect metastability (Marino,~1981).
	Synchronizers, the only traditional countermeasure, exponentially decrease the odds of maintained metastability over time.
	Trading synchronization delay for an increased probability to resolve metastability to logical 0 or~1, they do not guarantee success.

	We propose a fundamentally different approach:
	It is possible to \emph{contain} metastability by fine-grained logical masking so that it cannot infect the entire circuit.
	This technique \emph{guarantees} a limited degree of metastability in\dash---and uncertainty about\dash---the output.

	At the heart of our approach lies a time- and value-discrete model for metastability in synchronous clocked digital circuits.
	Metastability is propagated in a worst-case fashion, allowing to derive deterministic guarantees, without and unlike synchronizers.
	The proposed model permits positive results and passes the test of reproducing Marino's impossibility results.
	We fully classify which functions can be computed by circuits with standard registers.
	Regarding masking registers, we show that they become computationally strictly more powerful with each clock cycle, resulting in a non-trivial hierarchy of computable functions.

	Demonstrating the applicability of our approach, we present the first fault-tolerant distributed clock synchronization algorithm that deterministically guarantees correct behavior in the presence of metastability.
	As a consequence, clock domains can be synchronized without using synchronizers, enabling metastability-free communication between them.
\end{abstract}

\section{Introduction}
\label{sec:introduction}

A classic image invoked to explain metastability is a ball ``resting'' on the peak of a steep mountain.
In this unstable equilibrium the tiniest displacement exponentially self-amplifies, and the ball drops into a valley.
While for Sisyphus metastability admits some nanoseconds of respite, it fundamentally disrupts operation in \ac{VLSI} circuits by breaking the abstraction of Boolean logic.

In digital circuits, every bistable storage element can become \emph{metastable.}
Metastability refers to volatile states that usually involve an internal voltage strictly between logical $0$ and~$1$.
A metastable storage element can output deteriorated signals, e.g., voltages stuck between logical $0$ and logical~$1$, oscillations, late or unclean transitions, or otherwise unspecified behavior.
Such deteriorated signals may violate timing constraints or input specifications of gates and further storage elements.
Hence, deteriorated signals may spread through combinational logic and drive further bistables into metastability.
While metastability refers to a state of a bistable, we refer to the abovementioned deteriorated signals as ``metastable'' for the sake of exposition.

Unfortunately, any way of reading a signal from an unsynchronized clock domain or performing an analog-to-digital or time-to-digital conversion incurs the risk of a metastable result;
no physical implementation of a non-trivial digital circuit can deterministically avoid, resolve, or detect metastability~\cite{m-gtmo-81}.

Traditionally, the only countermeasure is to write a potentially metastable signal into a synchronizer~\cite{bg-ewbys-15,bgccz-mcsm-13,bgpdk-ds-10,g-fwfys-03,k-sads-08,kby-scp-02} and wait.
Synchronizers exponentially decrease the odds of maintained metastability over time~\cite{k-sads-08,kby-scp-02,v-bffspfr-80}:
In this unstable equilibrium the tiniest displacement exponentially self-amplifies and the bistable resolves metastability.
Put differently, the waiting time determines the probability to resolve to logical $0$ or~$1$.
Accordingly, this approach delays subsequent computations and does not guarantee success.

We propose a fundamentally different approach:
It is possible to \emph{contain} metastability by fine-grained logical masking so that it cannot infect the entire circuit.
This technique \emph{guarantees} a limited degree of metastability in\dash---and uncertainty about\dash---the output.
At the heart of our approach lies a model for metastability in synchronous clocked digital circuits.
Metastability is propagated in a worst-case fashion, allowing to derive deterministic guarantees, without and unlike synchronizers.

\paragraph*{The Challenge}

The problem with metastability is that it fundamentally disrupts operation in \ac{VLSI} circuits by breaking the abstraction of Boolean logic:
A metastable signal can neither be viewed as being logical $0$ or~$1$.
In particular, a metastable signal is not a random bit, and does not behave like an unknown but fixed Boolean signal.
As an example, the circuit that computes $\lnot x \lor x$ using a \gnot and a binary \gor gate may output an arbitrary signal value if $x$ is metastable:
$0$, $1$, or again a metastable signal.
Note that this is not the case for unknown, but Boolean,~$x$.
The ability of such signals to ``infect'' an entire circuit poses a severe challenge.

\paragraph*{The Status Quo}

The fact that metastability cannot be avoided, resolved or detected, the hazard of infecting entire circuits, and the unpleasant property of breaking the abstraction of Boolean logic have led to the predominant belief that waiting\dash---using well-designed synchronizers\dash---essentially is the \emph{only} method of coping with the threat of metastability:
Whenever a signal is potentially metastable, e.g., when it is communicated across a clock boundary, its value is written to a synchronizer.
After a predefined time, the synchronizer output is assumed to have stabilized to logical $0$ or~$1$, and the computation is carried out in classical Boolean logic.
In essence, this approach trades synchronization delay for increased reliability;
it does, however, not provide deterministic guarantees.

\paragraph*{Relevance}

\ac{VLSI} circuits grow in complexity and operating frequency, leading to a growing number unsynchronized clock domains, technology becomes smaller, and the operating voltage is decreased to save power~\cite{itrs-13}.
These trends increase the risk of metastable upsets.
Treating these risks in the traditional way\dash---by adding synchronizer stages\dash---increases synchronization delays and thus is counterproductive w.r.t.\ the desire for faster systems.
Hence, we urgently need alternative techniques to reliably handle metastability in both mission-critical and day-to-day systems.

\paragraph*{Our Approach}

We challenge this point of view and exploit that \emph{logical masking} provides some leverage.
If, e.g., one input of a \gnand gate is stable~$0$, its output remains $1$ even if its other input is arbitrarily deteriorated.
This is owed to the way gates are implemented in \ac{CMOS} logic and to transistor behavior under intermediate input voltage levels.

We conclude that it is possible to \emph{contain} metastability to a limited part of the circuit instead of attempting to resolve, detect, or avoid it altogether.
Given Marino's result~\cite{m-gtmo-81}, this is surprising, but not a contradiction.
More concretely, we show that a variety of operations can be performed in the presence of a limited degree of metastability in the input, maintaining an according guarantee on the output.

As an example, recall that in \ac{BRGC} $x$ and $x + 1$ always only differ in exactly one bit;
each upcount flips one bit.
Suppose \acp{ADC} output \ac{BRGC} but, due to their analog input, a possibly metastable bit $u$ decides whether to output $x$ or $x + 1$.
As $x$ and $x + 1$ only differ in a single bit, this bit is the only one that may become metastable in an appropriate implementation (the \acp{CMUX} discussed in Section~\ref{sec:mux}).
Hence, all possible stabilizations are in $\{x, x + 1\}$, we refer to this as \emph{precision-$1$.}
Among other things, we show that it is possible to sort such inputs in a way that the output still has precision-$1$.

We assume worst-case metastability propagation and still are able to \emph{guarantee} correct results.
This opens up an alternative to the classic approach of \emph{postponing} the actual computation by first using synchronizers.
Advantages over synchronizers are:
\begin{enumerate}
\item
	No time is lost waiting for (possible) stabilization.
	This permits fast response times as, e.g., useful for high-frequency clock synchronization in hardware, see Section~\ref{sec:arithmetic}.
	Note that this removes synchronization delay from the list of fundamental limits to the operating frequency.

\item
	Correctness is guaranteed deterministically instead of probabilistically.

\item
	Stabilization can, but is not required to, happen ``during'' the computation, i.e., synchronization and calculation happen simultaneously.
\end{enumerate}

\paragraph*{Separation of Concerns}

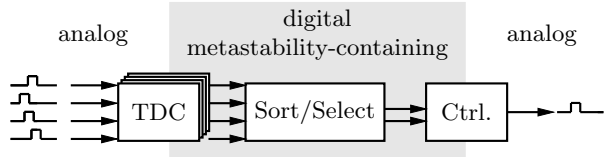
\begin{figure}
	\begin{center}
		{\small \begin{tikzpicture}[scale=2.54]
% dpic version 2016.01.11 option -g for TikZ and PGF 1.01
\ifx\dpiclw\undefined\newdimen\dpiclw\fi
\global\def\dpicdraw{\draw[line width=\dpiclw]}
\global\def\dpicstop{;}
\dpiclw=0.8bp
\dpiclw=0.8bp
\fill[fill=white](-0.269837,-0.719566) rectangle (0.14991,-0.419747)\dpicstop
\draw (-0.059964,-0.569656) node{TDC};
\dpicdraw[fill=white](-0.21844,-0.668168) rectangle (0.201307,-0.368349)\dpicstop
\dpicdraw[fill=white](-0.235572,-0.685301) rectangle (0.184175,-0.385482)\dpicstop
\dpicdraw[fill=white](-0.252705,-0.702433) rectangle (0.167042,-0.402614)\dpicstop
\dpicdraw[fill=white](-0.269837,-0.719566) rectangle (0.14991,-0.419747)\dpicstop
\draw (-0.059964,-0.569656) node{TDC};
\dpicdraw[fill=white](0.389765,-0.719566) rectangle (1.10933,-0.419747)\dpicstop
\draw (0.749548,-0.569656) node{Sort/Select};
\dpicdraw[fill=white](1.3252,-0.719566) rectangle (1.744947,-0.419747)\dpicstop
\draw (1.535073,-0.569656) node{Ctrl.};
\filldraw[line width=0bp](1.245248,-0.572987)
 --(1.3252,-0.552999)
 --(1.245248,-0.533012) --cycle\dpicstop
\dpicdraw (1.10933,-0.552999)
 --(1.306886,-0.552999)\dpicstop
\filldraw[line width=0bp](1.245248,-0.636283)
 --(1.3252,-0.616295)
 --(1.245248,-0.596307) --cycle\dpicstop
\dpicdraw (1.10933,-0.616295)
 --(1.306886,-0.616295)\dpicstop
\filldraw[line width=0bp](1.90485,-0.589644)
 --(1.984802,-0.569656)
 --(1.90485,-0.549668) --cycle\dpicstop
\dpicdraw (1.744947,-0.569656)
 --(1.966488,-0.569656)\dpicstop
\dpicdraw (-0.827261,-0.709572)
 --(-0.719326,-0.709572)\dpicstop
\dpicdraw (-0.719326,-0.709572)
 --(-0.719326,-0.661601)\dpicstop
\dpicdraw (-0.719326,-0.661601)
 --(-0.671355,-0.661601)\dpicstop
\dpicdraw (-0.671355,-0.661601)
 --(-0.671355,-0.709572)\dpicstop
\dpicdraw (-0.671355,-0.709572)
 --(-0.587405,-0.709572)\dpicstop
\filldraw[line width=0bp](-0.355545,-0.72956)
 --(-0.275594,-0.709572)
 --(-0.355545,-0.689584) --cycle\dpicstop
\dpicdraw (-0.515449,-0.709572)
 --(-0.293907,-0.709572)\dpicstop
\filldraw[line width=0bp](0.309813,-0.449729)
 --(0.389765,-0.429741)
 --(0.309813,-0.409753) --cycle\dpicstop
\dpicdraw (0.197881,-0.429741)
 --(0.371451,-0.429741)\dpicstop
\dpicdraw (-0.827261,-0.616295)
 --(-0.755304,-0.616295)\dpicstop
\dpicdraw (-0.755304,-0.616295)
 --(-0.755304,-0.568324)\dpicstop
\dpicdraw (-0.755304,-0.568324)
 --(-0.707333,-0.568324)\dpicstop
\dpicdraw (-0.707333,-0.568324)
 --(-0.707333,-0.616295)\dpicstop
\dpicdraw (-0.707333,-0.616295)
 --(-0.587405,-0.616295)\dpicstop
\filldraw[line width=0bp](-0.355545,-0.636283)
 --(-0.275594,-0.616295)
 --(-0.355545,-0.596307) --cycle\dpicstop
\dpicdraw (-0.515449,-0.616295)
 --(-0.293907,-0.616295)\dpicstop
\filldraw[line width=0bp](0.309813,-0.543006)
 --(0.389765,-0.523018)
 --(0.309813,-0.50303) --cycle\dpicstop
\dpicdraw (0.197881,-0.523018)
 --(0.371451,-0.523018)\dpicstop
\dpicdraw (-0.827261,-0.523018)
 --(-0.77929,-0.523018)\dpicstop
\dpicdraw (-0.77929,-0.523018)
 --(-0.77929,-0.475047)\dpicstop
\dpicdraw (-0.77929,-0.475047)
 --(-0.731319,-0.475047)\dpicstop
\dpicdraw (-0.731319,-0.475047)
 --(-0.731319,-0.523018)\dpicstop
\dpicdraw (-0.731319,-0.523018)
 --(-0.587405,-0.523018)\dpicstop
\filldraw[line width=0bp](-0.355545,-0.543006)
 --(-0.275594,-0.523018)
 --(-0.355545,-0.50303) --cycle\dpicstop
\dpicdraw (-0.515449,-0.523018)
 --(-0.293907,-0.523018)\dpicstop
\filldraw[line width=0bp](0.309813,-0.636283)
 --(0.389765,-0.616295)
 --(0.309813,-0.596307) --cycle\dpicstop
\dpicdraw (0.197881,-0.616295)
 --(0.371451,-0.616295)\dpicstop
\dpicdraw (-0.827261,-0.429741)
 --(-0.737315,-0.429741)\dpicstop
\dpicdraw (-0.737315,-0.429741)
 --(-0.737315,-0.38177)\dpicstop
\dpicdraw (-0.737315,-0.38177)
 --(-0.689344,-0.38177)\dpicstop
\dpicdraw (-0.689344,-0.38177)
 --(-0.689344,-0.429741)\dpicstop
\dpicdraw (-0.689344,-0.429741)
 --(-0.587405,-0.429741)\dpicstop
\filldraw[line width=0bp](-0.355545,-0.449729)
 --(-0.275594,-0.429741)
 --(-0.355545,-0.409753) --cycle\dpicstop
\dpicdraw (-0.515449,-0.429741)
 --(-0.293907,-0.429741)\dpicstop
\filldraw[line width=0bp](0.309813,-0.72956)
 --(0.389765,-0.709572)
 --(0.309813,-0.689584) --cycle\dpicstop
\dpicdraw (0.197881,-0.709572)
 --(0.371451,-0.709572)\dpicstop
\dpicdraw (2.00375,-0.569656)
 --(2.075707,-0.569656)\dpicstop
\dpicdraw (2.075707,-0.569656)
 --(2.075707,-0.521685)\dpicstop
\dpicdraw (2.075707,-0.521685)
 --(2.123678,-0.521685)\dpicstop
\dpicdraw (2.123678,-0.521685)
 --(2.123678,-0.569656)\dpicstop
\dpicdraw (2.123678,-0.569656)
 --(2.243606,-0.569656)\dpicstop
\draw (-0.399759,-0.166566) node{analog};
\draw (0.76254,-0.166566) node{\shortstack{digital\\%
metastability-containing}};
\draw (1.924838,-0.166566) node{analog};
\dpicdraw[fill=white,draw=white](-0.683787,-0.809511) rectangle (2.208867,0)\dpicstop
\dpicdraw[fill=black!10,draw=black!10](1.525079,-0.809511) rectangle (0,0)\dpicstop
\fill[fill=white](-0.269837,-0.719566) rectangle (0.14991,-0.419747)\dpicstop
\draw (-0.059964,-0.569656) node{TDC};
\dpicdraw[fill=white](-0.21844,-0.668168) rectangle (0.201307,-0.368349)\dpicstop
\dpicdraw[fill=white](-0.235572,-0.685301) rectangle (0.184175,-0.385482)\dpicstop
\dpicdraw[fill=white](-0.252705,-0.702433) rectangle (0.167042,-0.402614)\dpicstop
\dpicdraw[fill=white](-0.269837,-0.719566) rectangle (0.14991,-0.419747)\dpicstop
\draw (-0.059964,-0.569656) node{TDC};
\dpicdraw[fill=white](0.389765,-0.719566) rectangle (1.10933,-0.419747)\dpicstop
\draw (0.749548,-0.569656) node{Sort/Select};
\dpicdraw[fill=white](1.3252,-0.719566) rectangle (1.744947,-0.419747)\dpicstop
\draw (1.535073,-0.569656) node{Ctrl.};
\filldraw[line width=0bp](1.245248,-0.572987)
 --(1.3252,-0.552999)
 --(1.245248,-0.533012) --cycle\dpicstop
\dpicdraw (1.10933,-0.552999)
 --(1.306886,-0.552999)\dpicstop
\filldraw[line width=0bp](1.245248,-0.636283)
 --(1.3252,-0.616295)
 --(1.245248,-0.596307) --cycle\dpicstop
\dpicdraw (1.10933,-0.616295)
 --(1.306886,-0.616295)\dpicstop
\filldraw[line width=0bp](1.90485,-0.589644)
 --(1.984802,-0.569656)
 --(1.90485,-0.549668) --cycle\dpicstop
\dpicdraw (1.744947,-0.569656)
 --(1.966488,-0.569656)\dpicstop
\dpicdraw (-0.827261,-0.709572)
 --(-0.719326,-0.709572)\dpicstop
\dpicdraw (-0.719326,-0.709572)
 --(-0.719326,-0.661601)\dpicstop
\dpicdraw (-0.719326,-0.661601)
 --(-0.671355,-0.661601)\dpicstop
\dpicdraw (-0.671355,-0.661601)
 --(-0.671355,-0.709572)\dpicstop
\dpicdraw (-0.671355,-0.709572)
 --(-0.587405,-0.709572)\dpicstop
\filldraw[line width=0bp](-0.355545,-0.72956)
 --(-0.275594,-0.709572)
 --(-0.355545,-0.689584) --cycle\dpicstop
\dpicdraw (-0.515449,-0.709572)
 --(-0.293907,-0.709572)\dpicstop
\filldraw[line width=0bp](0.309813,-0.449729)
 --(0.389765,-0.429741)
 --(0.309813,-0.409753) --cycle\dpicstop
\dpicdraw (0.197881,-0.429741)
 --(0.371451,-0.429741)\dpicstop
\dpicdraw (-0.827261,-0.616295)
 --(-0.755304,-0.616295)\dpicstop
\dpicdraw (-0.755304,-0.616295)
 --(-0.755304,-0.568324)\dpicstop
\dpicdraw (-0.755304,-0.568324)
 --(-0.707333,-0.568324)\dpicstop
\dpicdraw (-0.707333,-0.568324)
 --(-0.707333,-0.616295)\dpicstop
\dpicdraw (-0.707333,-0.616295)
 --(-0.587405,-0.616295)\dpicstop
\filldraw[line width=0bp](-0.355545,-0.636283)
 --(-0.275594,-0.616295)
 --(-0.355545,-0.596307) --cycle\dpicstop
\dpicdraw (-0.515449,-0.616295)
 --(-0.293907,-0.616295)\dpicstop
\filldraw[line width=0bp](0.309813,-0.543006)
 --(0.389765,-0.523018)
 --(0.309813,-0.50303) --cycle\dpicstop
\dpicdraw (0.197881,-0.523018)
 --(0.371451,-0.523018)\dpicstop
\dpicdraw (-0.827261,-0.523018)
 --(-0.77929,-0.523018)\dpicstop
\dpicdraw (-0.77929,-0.523018)
 --(-0.77929,-0.475047)\dpicstop
\dpicdraw (-0.77929,-0.475047)
 --(-0.731319,-0.475047)\dpicstop
\dpicdraw (-0.731319,-0.475047)
 --(-0.731319,-0.523018)\dpicstop
\dpicdraw (-0.731319,-0.523018)
 --(-0.587405,-0.523018)\dpicstop
\filldraw[line width=0bp](-0.355545,-0.543006)
 --(-0.275594,-0.523018)
 --(-0.355545,-0.50303) --cycle\dpicstop
\dpicdraw (-0.515449,-0.523018)
 --(-0.293907,-0.523018)\dpicstop
\filldraw[line width=0bp](0.309813,-0.636283)
 --(0.389765,-0.616295)
 --(0.309813,-0.596307) --cycle\dpicstop
\dpicdraw (0.197881,-0.616295)
 --(0.371451,-0.616295)\dpicstop
\dpicdraw (-0.827261,-0.429741)
 --(-0.737315,-0.429741)\dpicstop
\dpicdraw (-0.737315,-0.429741)
 --(-0.737315,-0.38177)\dpicstop
\dpicdraw (-0.737315,-0.38177)
 --(-0.689344,-0.38177)\dpicstop
\dpicdraw (-0.689344,-0.38177)
 --(-0.689344,-0.429741)\dpicstop
\dpicdraw (-0.689344,-0.429741)
 --(-0.587405,-0.429741)\dpicstop
\filldraw[line width=0bp](-0.355545,-0.449729)
 --(-0.275594,-0.429741)
 --(-0.355545,-0.409753) --cycle\dpicstop
\dpicdraw (-0.515449,-0.429741)
 --(-0.293907,-0.429741)\dpicstop
\filldraw[line width=0bp](0.309813,-0.72956)
 --(0.389765,-0.709572)
 --(0.309813,-0.689584) --cycle\dpicstop
\dpicdraw (0.197881,-0.709572)
 --(0.371451,-0.709572)\dpicstop
\dpicdraw (2.00375,-0.569656)
 --(2.075707,-0.569656)\dpicstop
\dpicdraw (2.075707,-0.569656)
 --(2.075707,-0.521685)\dpicstop
\dpicdraw (2.075707,-0.521685)
 --(2.123678,-0.521685)\dpicstop
\dpicdraw (2.123678,-0.521685)
 --(2.123678,-0.569656)\dpicstop
\dpicdraw (2.123678,-0.569656)
 --(2.243606,-0.569656)\dpicstop
\draw (-0.399759,-0.166566) node{analog};
\draw (0.76254,-0.166566) node{\shortstack{digital\\%
metastability-containing}};
\draw (1.924838,-0.166566) node{analog};
\end{tikzpicture}}
	\end{center}
	\caption{%
		The separation of concerns (analog -- digital metastability-containing -- analog) for fault-tolerant clock synchronization in hardware.}
	\label{fig:clocksync}
\end{figure}

Clearly, the impossibility of resolving metastability still holds;
metastability may still occur, even if it is contained.
Hence, a \emph{separation of concerns,} compare Figure~\ref{fig:clocksync}, is key to our approach.

For the purpose of illustration, consider a hardware clock-synchronization algorithm, we discuss this in Section~\ref{sec:arithmetic}.
We start in the \emph{analog} world: nodes generate clock pulses.
Each node measures the time differences between its own and all other nodes' pulses using \acp{TDC}.
Since this involves entering the \emph{digital} world, metastability in the measurements is unavoidable~\cite{m-gtmo-81}.
The traditional approach is to hold the \ac{TDC} outputs in synchronizers, spending time and thus imposing a limit on the operating frequency.
But as discussed above, it is possible to limit the metastability of each measurement to at most one bit in \acs{BRGC}-encoded numbers, where the metastable bit represents the ``uncertainty between $x$ and $x + 1$ clock ticks,'' i.e., precision-$1$.

We apply \emph{metastability-containing} components to digitally process these inputs to derive digital correction parameters for the node's oscillator.
These parameters contain at most one metastable bit, as above accounting for precision-$1$.
We convert them to an \emph{analog} control signal for the oscillator.
This way, the metastability translates to a small frequency offset within the uncertainty from the initial \ac{TDC} measurements.

In short, metastability is introduced at the \acp{TDC}, \emph{deterministically} contained in the digital subcircuit, and ultimately absorbed by the analog control signal.

\subsection{Our Contribution}
\label{sec:introduction-contribution}

In Section~\ref{sec:model}, we present a rigorous time-discrete value-discrete model for metastability in clocked as well as in purely combinational digital circuits.
We consider two types of registers: simple (standard) registers that do not provide any guarantees regarding metastability and masking registers that can ``hide'' internal metastability to some degree using high- or low-threshold inverters.
The propagation of metastability is modeled in a worst-case fashion and metastable registers may or may not stabilize to $0$ or~$1$.
Hence, the resulting model thus allows us to derive deterministic guarantees concerning circuit behavior under metastable inputs.

We consider the model that allows a novel and fundamentally different worst-case treatment of metastability our main contribution.
Accordingly, we are obligated to demonstrate that the model is not too pessimistic, i.e., that it allows non-trivial positive results.
We do this in Section~\ref{sec:mux}, where we develop \acp{CMUX}, these also serve as an example for the concept of metastability-containment as a whole.
At the same time, we are obligated to verify that it properly reflects the physical behavior of digital circuits, i.e., that it is sufficiently pessimistic.
We establish some basic properties in Section~\ref{sec:basics} and continue with a reality check in Section~\ref{sec:realitycheck}, showing that the physical impossibility of avoiding, resolving, or detecting metastability~\cite{m-gtmo-81} holds in our model.

Having established some confidence that our model properly reflects the physical world and allows reasoning about circuit design, we turn our attention to the question of computability.
In Section~\ref{sec:hierarchy}, we analyze what functions are computable by circuits w.r.t.\ the available register types and the number of clock cycles.
Let $\Fun_M^r$ denote the class of functions that can be implemented by an arbitrary circuit in $r$ clock cycles;\footnote{
	The $M$ indicates that the circuit may comprise masking registers.
} analogously, let $\Fun_S^r$ denote the class of functions implementable in $r$ clock cycles of circuits that can only use simple registers.
We show that the number of clock cycles is irrelevant for combinational and simple circuits:
\begin{equation}
	\cdots = \Fun_S^2 = \Fun_S^1 = \Fun_M^1 \subsetneq \Fun_M^2 \subsetneq \cdots\,.
\end{equation}
The collapse of the hierarchy $\Fun_S^r$ reflects the intuition from electrical engineering that synchronous Boolean circuits can be unrolled.
In the presence masking registers, however, unrolling does not yield equivalent circuits and we obtain a strict inclusion.

In Section~\ref{sec:simple}, we move on to demonstrating that even with simple registers, non-trivial functions can be computed in the face of worst-case propagation of metastability.
To this end, we fully classify $\Fun_S$.
Furthermore, we establish the \emph{metastable closure,} the strictest possible extension of a function specification that allows it to be computed by a combinational or simple circuit.
Our classification provides an extremely simple test deciding whether a desired specification can be implemented.

Finally, we apply our techniques to show that an advanced, useful circuit is in reach.
We show in Section~\ref{sec:arithmetic} that all operations required by the widely used~\cite{bbb+-fcp-03,kb-tta-03} fault-tolerant clock synchronization algorithm of Lundelius Welch and Lynch~\cite{ll-ftacs-88}\dash---$\max$ and $\min$, sorting, and conversion between \ac{TC} and \ac{BRGC}\dash---can be performed in a metastability-containing manner.
Employing the abovementioned separation of concerns, a hardware implementation of the entire algorithm is within reach, providing the deterministic guarantee that the algorithm works correctly at all times, despite metastable upsets originating in the \acp{TDC} and without synchronizers.

As a consequence, we show that
\begin{enumerate*}
\item
	synchronization delay poses no fundamental limit on the operating frequency of clock synchronization in hardware and that

\item
	clock domains can be synchronized without synchronizers.
\end{enumerate*}
The latter shows that we may eliminate communication across unsynchronized clock domains as a source of metastable upsets altogether.

\subsection{Related Work}
\label{sec:introduction-rw}

\paragraph*{Metastability}

The phenomenon of metastable signals in fact has been studied for decades~\cite{k-sads-08} with the following key results.
\begin{enumerate*}
\item
	No physical implementation of a digital circuit can reliably avoid, resolve, or detect metastability;
	any digital circuit, including ``detectors,'' producing different outputs for different input signals can be forced into metastability~\cite{m-gtmo-81}.

\item
	The probability of an individual event generating metastability can be kept low.
	Large transistor counts and high operational frequencies, low supply voltages, temperature effects, and changes in technology, however, disallow to neglect the problem~\cite{bgccz-mcsm-13}.

\item
	Being an unstable equilibrium, the probability that, e.g., a memory cell remains in a metastable state decreases exponentially over time~\cite{k-sads-08,kby-scp-02,v-bffspfr-80}.
	Thus, waiting for a sufficiently long time reduces the probability of sustained metastability to within acceptable bounds.
\end{enumerate*}

\paragraph*{Synchronizers}

The predominant technique to cope with metastable upsets is to use synchronizers~\cite{bg-ewbys-15,bgccz-mcsm-13,bgpdk-ds-10,g-fwfys-03,k-sads-08,kby-scp-02}.
Synchronizers are carefully designed~\cite{bg-ewbys-15,g-fwfys-03} bistable storage elements that hold potentially metastable signals, e.g., after communicating them across a clock boundary.
After a predefined time, the synchronizer output is assumed to have stabilized to logical $0$ or $1$ and the computation is carried out in classical Boolean logic.
In essence, this approach trades delay for increased reliability, typically expressed as \ac{MTBF}
\begin{equation}\label{eq:mtbf}
	\text{\acs{MTBF}} = \frac{e^{t / \tau}}{T_W F_C F_D},
\end{equation}
where $F_C$ and $F_D$ are the clock and data transition frequencies, $\tau$~and $T_W$ are technology-dependent values, and $t$ is the predetermined time allotted for synchronization~\cite{bg-ewbys-15,bgccz-mcsm-13,bgpdk-ds-10,g-fwfys-03,k-sads-08,kby-scp-02,ty-rtlmhcdcl-12,tym-eslusl-14}.
Synchronizers, however, do not provide deterministic guarantees and avoiding synchronization delay is an important issue~\cite{ty-rtlmhcdcl-12,tym-eslusl-14}.

\paragraph*{Glitch/Hazard Propagation}

Logically masking metastability is related to glitch-free and hazard-free circuits.

Metastability-containing circuits are related to glitch-free/hazard-free circuits, which have been extensively studied since Huffman~\cite{h-duhfsn-57} and Unger~\cite{u-hdassc-59} introduced them.
Eichelberger~\cite{e-hdcssc-65} extended these results to multiple switching inputs and dynamic hazards, Brzozowski and Yoeli extended the simulation algorithm~\cite{by-tmgn-79}, Brzozowski et~al.\ surveyed techniques using higher-valued logics~\cite{bei-ahd-01} such as Kleene's $3$-valued extension of Boolean logic, and Mendler et~al.\ studied delay requirements needed to achieve consistency with simulated results~\cite{msb-cbcetts-12}.

While we too resort to Kleene's $3$-valued to model metastability, there are differences to the classical work on hazard-tolerant circuits:
\begin{enumerate*}
\item
	A common assumption in hazard detection is that inputs only perform well-defined, clean transitions, i.e., the assumption of a hazard-free input-generating circuitry is made.
	This is the key difference to metastability-containment:
	Metastability encompasses much more than inputs that are in the process of switching;
	metastable signals may or may not be in the process of completing a transition, may be oscillating, and may get ``stuck'' at an intermediate voltage.

\item
	Another common assumption in hazard detection is that circuits have a constant delay.
	This is no longer the case in the presence of metastability;
	unless metastability is properly masked, circuit delays can deteriorate in the presence of metastable input signals, even if the circuit eventually generates a stable output~\cite{fk-emcm-17}.
	This can cause late transitions that potentially drive further registers into metastability.

\item
	Glitch-freedom is no requirement for metastability-containment.

\item
	When studying synthesis, we allow for specifications where outputs may contain metastable bits.
	This is necessary for non-trivial specifications in the presence of metastable inputs~\cite{m-gtmo-81}.

\item
	We allow a circuit to compute a function in multiple clock cycles.

\item
	Circuits may comprise masking registers~\cite{k-sads-08}.
\end{enumerate*}

\paragraph*{OR Causality}

The work on weak~(OR) causality in asynchronous circuits~\cite{ykkl-orcmhi-94} studies the computation of functions under availability of only a proper subset of its parameters.
As an example, consider a Boolean function $f(x,y)$, where $f(0,0) = f(0,1)$.
An early-deciding asynchronous module may set its output as soon as $x = 0$ arrives at its input, disregarding the value of~$y$.
Early-deciding circuits, however, differ from our work because they are neither clocked synchronous designs nor do they necessarily operate correctly in presence of metastable input bits: $f(0,\meta) = f(0,0) = f(0,1)$ does not necessarily hold.

\paragraph*{Speculative Computing}

To the best of our knowledge, the most closely related work is that by Tarawneh et~al.\ on speculative computing~\cite{ty-rtlmhcdcl-12,tym-eslusl-14}.
The idea is the following:
When computing $f(x,y)$ in presence of a potentially metastable input bit~$x$,
\begin{enumerate*}
\item
	speculatively compute both $f(0,y)$ and $f(1,y)$,

\item
	in parallel, store the input bit $x$ in a synchronizer for a predefined time that provides a sufficiently large probability of resolving metastability of~$x$, and

\item
	use $x$ to select whether to output $f(0,y)$ or $f(1,y)$.
\end{enumerate*}
This hides (part of) the delay needed to synchronize~$x$.

Like our approach, speculative computations allow for an overlap of synchronization and computation time.
The key differences are:
\begin{enumerate*}
\item
	Relying on synchronizers, speculative computing incurs a non-zero probability of failure;
	metastability-containment insists on deterministic guarantees.

\item
	In speculative computing, the set of potentially metastable bits $X$ must be known in advance.
	Regardless of the considered function, the complexity of a speculative circuit grows exponentially in~$|X|$.
	Neither is the case for metastability-containment, as illustrated by several circuits~\cite{blm-nomcsn-17,fklp-mametdc-17,lm-emcgc2s-16,tfl-mtc-17}.

\item
	Our model is rooted in an extension of Boolean logic, i.e., uses a different function space.
	Hence, we face the question of computability of such functions by digital circuits;
	this question does not apply to speculative computing as it uses traditional Boolean functions.
\end{enumerate*}

\paragraph*{Metastability-Containing Circuits}

Many of the proposed techniques have been successfully employed to obtain metastability-aware \acp{TDC}~\cite{fklp-mametdc-17}, metastability-containing \ac{BRGC} sorting networks~\cite{blm-nomcsn-17,lm-emcgc2s-16}, \acp{CMUX}~\cite{fk-emcm-17}, and metastability-tolerant network-on-chip routers~\cite{tfl-mtc-17}.
Simulations verify the positive impact of metastability-containing techniques~\cite{blm-nomcsn-17,fk-emcm-17,tfl-mtc-17}.
Most of these works channel efforts towards metastability-containing \ac{FPGA} and \ac{ASIC} implementations of fault-tolerant distributed clock synchronization;
this paper establishes that all required components are within reach.

\subsection{Notation}
\label{sec:introduction-notation}

$\N_0$ and $\N$ denote the natural numbers with and without~$0$.
We abbreviate $[k] := \{ \ell \in \N_0 \mid \ell < k \}$ for $k \in \N_0$.
% [0] = \emptyset needed for circuit definition
Tuples $a, b$ are concatenated by $a \circ b$, and given a set~$S$, $\Pow(S) := \{ S' \subseteq S \}$ is its power set.

\section{Model of Computation}
\label{sec:model}

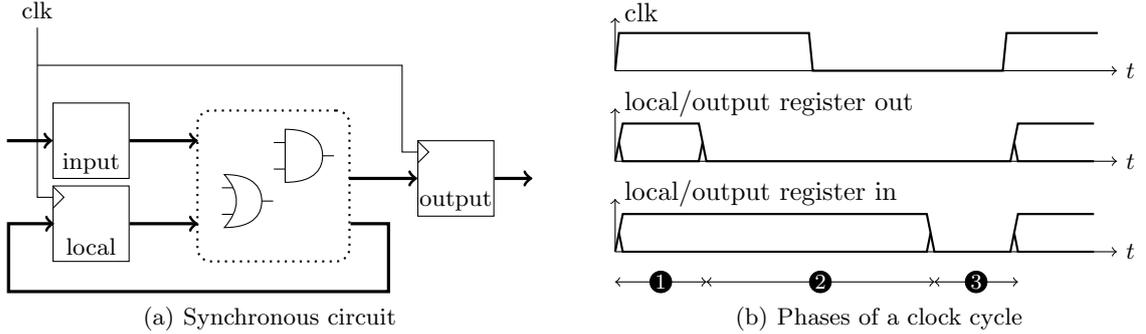
\begin{figure}
\begin{center}
\hfill\subfigure[Synchronous circuit]{
\begin{tikzpicture}[scale=1,transform shape,cktbaselength=0.5pt]

  % local register
  \path (-1,-0.1) node[draw,rectangle,minimum width=1cm,minimum height=1cm] (loc) {} -- ++(0,-0.3)
    node {\small local};
  \draw (loc.west)++(0,0.2) -- ++(0.15,0.15) -- ++(-0.15,0.15);

  % in register
  \path (-1,1) node[draw,rectangle,minimum width=1cm,minimum height=1cm] (in) {} -- ++(0,-0.3)
    node {\small input};
  %\draw (in.west)++(0,0.2) -- ++(0.15,0.15) -- ++(-0.15,0.15);

  % logic
  \draw (1.0,0.2) node[or2,draw,rotate=90] (o1) {};
  \draw (1.8,0.8) node[and2,draw,rotate=90] (o2) {};
  \draw (o1.a) -- ++(-0.15,0);
  \draw (o1.b) -- ++(-0.15,0);
  \draw (o1.z) -- ++(0.15,0);
  \draw (o2.a) -- ++(-0.15,0);
  \draw (o2.b) -- ++(-0.15,0);
  \draw (o2.z) -- ++(0.15,0);
  \draw (1.4,0.4) node[draw,thick,rounded corners=5pt,rectangle,dotted,minimum width=2cm,minimum height=2cm] (logic) {};

  % out register
  \path (3.8,0.5) node[draw,rectangle,minimum width=1cm,minimum height=1cm] (out) {} -- ++(0,-0.3)
    node {\small output};
  \draw (out.west)++(0,0.2) -- ++(0.15,0.15) -- ++(-0.15,0.15);

  % interconnect
  \draw[very thick,->] (in.east) -- ++(0.9,0);
  \draw[very thick,->] (in.west)++(-0.6,0) -- ++(0.6,0);
  \draw[very thick,->] (loc.east) -- ++(0.9,0);
  \draw[very thick,->] (out.west)++(-0.9,0) -- ++(0.9,0);
  \draw[very thick,->] (out.east) -- ++(0.5,0);
  \draw[very thick,->] (logic.east)++(0,-0.5) -- ++(0.5,0) -- ++(0,-0.9) -- ++(-5.0,0) -- ++(0,0.9) -- ++(0.6,0);

  % clock
  %\draw[-,thin] (in.west)++(-0.2,1.2) node[above] {\small clk} -- ++(0,-0.85) -- ++(0.2,0);
  \draw[-,thin] (loc.west)++(-0.2,2.6) node[above] {\small clk} -- ++(0,-2.25) -- ++(0.2,0);
  \draw[-,thin] (out.west)++(-5,1.5) -- ++(4.8,0) -- ++(0,-1.15) -- ++(0.2,0);

\end{tikzpicture}
\label{fig:EE.a}}
\hfill\subfigure[Phases of a clock cycle]{
\begin{tikzpicture}[scale=1,transform shape,cktbaselength=0.5pt]
  \draw[->] (0,1.2) -- ++(6.6,0) node[right] {\small $t$};
  \draw[->] (0,1.2) -- ++(0,0.7) node[right,yshift=2pt] {clk};

  \draw[->] (0,0) -- ++(6.6,0) node[right] {\small $t$};
  \draw[->] (0,0) -- ++(0,0.7) node[right,yshift=2pt] {local/output register out};

  \draw[->] (0,-1.2) -- ++(6.6,0) node[right] {\small $t$};
  \draw[->] (0,-1.2) -- ++(0,0.7) node[right,yshift=2pt] {local/output register in};

  % clk signal
  \draw[-,thick] (0,1.2) --
  ++(0.05,0.5) -- ++(2.5,0) --
  ++(0.05,-0.5) -- ++(2.5,0) --
  ++(0.05,0.5) -- ++(1.2,0);

  % reg out signal
  \draw[-,thick] (0,0.0) --
  ++(0.1,0.5) -- ++(1,0) --
  ++(0.1,-0.5) -- ++(4,0) --
  ++(0.1,0.5) -- ++(1,0);

  \draw[-,thick] (0.05,0.25) --
  ++(0.05,-0.25) -- ++(1,0) --
  ++(0.05,0.25);

  \draw[-,thick] (5.25,0.25) --
  ++(0.05,-0.25) -- ++(1,0);

  % reg in signal
  \draw[-,thick] (0,-1.2) --
  ++(0.1,0.5) -- ++(4,0) --
  ++(0.1,-0.5) -- ++(1,0) --
  ++(0.1,0.5) -- ++(1,0);

  \draw[-,thick] (0.05,-0.95) --
  ++(0.05,-0.25) -- ++(4,0) --
  ++(0.05,0.25);

  \draw[-,thick] (5.25,-0.95) --
  ++(0.05,-0.25) -- ++(1,0);

  % below the phases
  \draw[<->] (0,-1.6) -- ++(1.2,0) node[midway,draw,fill=black,circle,text=white,inner sep=0.5pt]{\footnotesize 1};
  \draw[<->] (1.2,-1.6) -- ++(3.0,0) node[midway,draw,fill=black,circle,text=white,inner sep=0.5pt]{\footnotesize 2};
  \draw[<->] (4.2,-1.6) -- ++(1.1,0) node[midway,draw,fill=black,circle,text=white,inner sep=0.5pt]{\footnotesize 3};

\end{tikzpicture}
\label{fig:EE.b}}\hfill
\end{center}
\caption{Generic synchronous state machine design in \subref{fig:EE.a}. The input register is initially prefilled.
  Local and output registers are updated at each rising clock transition. The circuit behavior over time is depicted in \subref{fig:EE.b}.
  The three phases of a clock cycle are shown:
  (1)~register output stabilization, (2)~propagation of outputs through combinational logic to register inputs, and~(3)
  stable register inputs.}
\label{fig:EEcircuit}
\end{figure}

We propose a time-discrete and value-discrete model in which registers can become metastable and their resulting output signals deteriorated.
The model supports synchronous, clocked circuits composed of registers and combinational logic and purely combinational circuits.
Specifically, we study the generic synchronous state-machine design depicted in Figure~\ref{fig:EEcircuit}.
Data is initially written into input registers.
At each rising clock transition, local and output registers update their state according to the circuit's combinational logic.
Figure~\ref{fig:EE.b} shows the circuit's behavior over time:
\begin{enumerate*}
\item\label{enum:phase1}
	During the first phase, the output of the recently updated local and output
  registers stabilizes.
	This is accounted for by the \emph{clock-to-output} time that can be bounded, except for the case of a metastable register.
In this case, no deterministic upper bound exists.
\item\label{enum:phase2} During phase two, the stable register output propagates through the combinational logic to the register inputs.
Its duration can be upper-bounded by the worst-case propagation delay through the combinational part.
\item\label{enum:phase3} In the third phase, the register inputs are stable, ready to be read (sampled), and result in updated local and output register states.
The duration of this phase is chosen such that it can account for potential delays in phase~\ref{enum:phase1};
this can mitigate some metastable upsets.
If the stabilization in phase~\ref{enum:phase1}, however, also exceeds the additional time in phase~\ref{enum:phase3}, a register may read an unstable input value,
  potentially resulting in a metastable register.
\end{enumerate*}

As motivated, metastable registers output an undefined, arbitrarily deteriorated signal. Deteriorated can mean any constant voltage between logical $0$ and logical~$1$, arbitrary signal behavior over time, oscillations, or simply violated timing constraints, such as late signal transitions.
Furthermore, deteriorated signals can cause registers to become metastable, e.g., due to violated constraints regarding timing or input voltage.
Knowing full well that metastability is a state of a \emph{bistable} element and not a signal value or voltage, we still need to talk about the ``deterioration caused by or potentially causing metastability in a register'' in \emph{signals.}
For the sake of presentation\dash---and as these effects are causally linked\dash---we refer to both phenomena using the term metastability without making the distinction explicit.

Our model uses Kleene's $3$-valued logic, a ternary extension of binary logic;
the third value appropriately expresses the uncertainty about gate behavior in the presence of metastability.
In the absence of metastability, our model behaves like a traditional, deterministic, binary circuit model.
In order to obtain \emph{deterministic} guarantees, we assume worst-case propagation of metastability:
If a signal can be ``infected'' by metastability, there is no way to prevent that.

Section~\ref{sec:mux} demonstrates our model using \acp{CMUX}, and Section~\ref{sec:realitycheck} ensures that it is not ``too optimistic'' by proving that it reproduces well-known impossibility results.
Concretely, we show that for circuits in our model avoiding, detecting, and resolving metastability is impossible, just as in physical circuits~\cite{m-gtmo-81}.
Clearly, this obliges us to provide evidence that our model has practical relevance, i.e., that it is indeed possible to perform meaningful computations.
Surprisingly, the classification derived in Section~\ref{sec:simple} entails that many interesting functions can be implemented by circuits, which is discussed in Section~\ref{sec:arithmetic}.

In our model circuits are synchronous state machines:
Combinational logic, represented by gates, maps a circuit state to possible successor states.
The combinational logic uses, and registers store, \emph{signal values} $\BM := \{0, 1, \meta\}$.
\meta~represents a \emph{metastable} signal, the only source of non-determinism.
The classical \emph{stable} Boolean signal values are $\B := \{0, 1\}$.
Let $x \in \BM^k$ be a $k$-bit tuple.
Stored in registers over time, the metastable bits may resolve to $0$ or~$1$.
The set of partial resolutions of $x$ is $\ResM(x)$, and the set of metastability-free, i.e., completely stabilized, resolutions is $\Res(x)$.
If $m$ bits in $x$ are metastable, $|\ResM(x)| = 3^m$ and $|\Res(x)| = 2^m$, since \meta serves as ``wildcard'' for $\BM$ and~$\B$, respectively.
Formally,
\begin{align}
	\ResM(x)
		&:= \left\{ y \in \BM^k \mid \forall i \in [k]\colon x_i = y_i \lor x_i = \meta \right\}, \\
	\Res(x)
		&:= \ResM(x) \cap \B^k. \label{eq:resm}
\end{align}
%\begin{observation}\label{obs:resm}
%	For $x' \in \ResM(x)$, $\ResM(x') \subseteq \ResM(x)$.
%\end{observation}

\subsection{Registers}
\label{sec:model-registers}

\begin{figure}
	\hfill\subfigure[simple]{
		\def\svgwidth{.25\linewidth}
		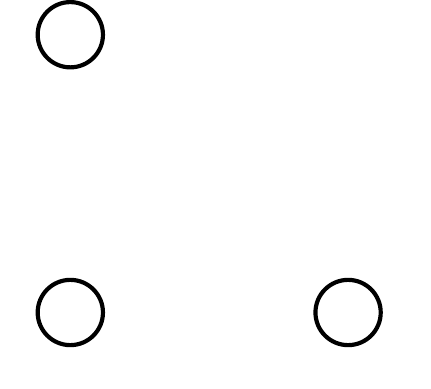
		\label{fig:register-states-simple}
	}\hfill
	\subfigure[mask-$0$]{
		\def\svgwidth{.25\linewidth}
		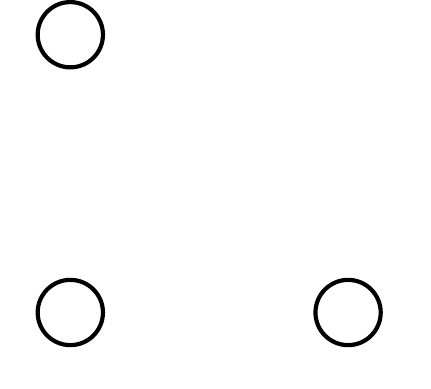
		\label{fig:register-states-mask0}
	}\hfill
	\subfigure[mask-$1$]{
		\def\svgwidth{.25\linewidth}
		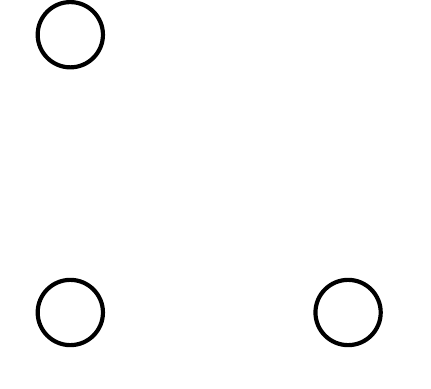
		\label{fig:register-states-mask1}
	}\hfill{}
	\caption{%
		Registers as non-deterministic state machines;
		state transitions represent reads and are associated with an output.
		As we propose a worst-case model, the dashed state transitions can be left out.
	}
	\label{fig:register-states}
\end{figure}

We consider three types of single-bit registers, all of which behave just like in binary circuit models unless metastability occurs:
\begin{enumerate*}
\item
	\emph{simple} registers which are oblivious to metastability, and

\item
	registers that mask an internal metastable state to an output of~$1$ \emph{(mask-$1$)} or

\item
	to~$0$ \emph{(mask-$0$)}.
\end{enumerate*}
Physical realizations of masking registers are obtained by flip-flops with high- or low-threshold inverters at the output, amplifying an internal metastable signal to $1$ or~$0$; see, e.g., Section~3.1 on metastability filters in~\cite{k-sads-08}.
A register $R$ has a \emph{type} (simple, mask-$0$, or mask-$1$) and a state $x_R \in \BM$.
$R$~behaves according to $x_R$ and its type's non-deterministic state machine in Figure~\ref{fig:register-states}.
Each clock cycle, $R$~performs one state transition annotated with some $o_R \in \BM$, which is the result of sampling $R$ at that clock cycle's rising clock flank.
This happens exactly once per clock cycle in our model and we refer to it as \emph{reading}~$R$.
The state transitions are not caused by sampling $R$ but account for the possible resolution of metastability during the preceeding clock cycle.

Consider a simple register in Figure~\ref{fig:register-states-simple}.
When in state~$0$, its output and successor state are both~$0$; it behaves symmetrically in state~$1$.
In state~\meta, however, any output in $\BM$ combined with any successor state in $\BM$ is possible.

Since our goal is to design circuits that operate correctly under metastability even if it never resolves, we make two pessimistic simplifications:
\begin{enumerate*}
\item
	If there are three parallel state transitions from state $x$ to $x'$ with outputs $0$,\,$1$,\,\meta, we only keep the one with output~\meta, and

\item
	if, for some fixed output $o \in \BM$, there are state transitions from a state $x$ to multiple states including~\meta, we only keep the one with successor state~\meta.
\end{enumerate*}
This simplification is obtained by ignoring the dashed state transitions in Figure~\ref{fig:register-states}, and we maintain it throughout the paper.
Observe that the dashed lines are a remnant of the highly non-deterministic ``anything can happen'' behavior in the physical world;
if one is pessimistic about the behavior, however, one obtains the proposed simplification that ignores the dashed state transitions.

The mask-$b$ registers, $b \in \B$, shown in Figures~\ref{fig:register-states-mask0} and~\ref{fig:register-states-mask1}, exhibit the following behavior:
As long as their state remains~\meta, they output $b \neq \meta$;
only when their state changes from \meta to $1 - b$ they output \meta once, after that they are stable.

\subsection{Gates}
\label{sec:model-gates}

We model the behavior of combinational gates in the presence of metastability.
A \emph{gate} is defined by $k \in \N_0$ input ports, one output port\dash---gates with $k \geq 2$ distinct output ports are represented by $k$ single-output gates\dash---and a Boolean function $f\colon \B^k \to \B$.
We generalize $f$ to $f_\meta\colon \BM^k \to \BM$ as follows.
Each metastable input can be perceived as~$0$, as~$1$, or as metastable superposition~\meta.
Hence, to determine $f_\meta(x)$, consider $O := \{ f(x') \mid x' \in \Res(x) \}$, the set of possible outputs of $f$ after $x$ fully stabilized.
If there is only a single possible output, i.e., $O = \{b\}$ for some $b \in \B$, the metastable bits in $x$ have no influence on $f(x)$ and we set $f_\meta(x) := b$.
Otherwise, $O = \B$, i.e., the metastable bits can change $f(x)$, and we set $f_\meta(x) := \meta$.
Observe that this is equivalent to Kleene's $3$-valued logic and that $f_\meta(x) = f(x)$ for all $x \in \B^k$.

%\medskip{
\begin{table}[ht]
	\hfill
	\begin{tabular}{c||c|c}
		$f^{\gand}$ & $0$ & $1$ \\
		\hline
		\hline
		$0$         & $0$ & $0$ \\
		\hline
		$1$         & $0$ & $1$
	\end{tabular}
	\hfill
	\begin{tabular}{c||c|c|c}
		$f_\meta^{\gand}$ & $0$ & $1$ & $\meta$ \\
		\hline
		\hline
		$0$     & $0$ &     $0$ &     $0$ \\
		\hline
		$1$     & $0$ &     $1$ & $\meta$ \\
		\hline
		$\meta$ & $0$ & $\meta$ & $\meta$
	\end{tabular}
	\hfill
	\begin{tabular}{c||c|c}
		$f^{\gor}$ & $0$ & $1$ \\
		\hline
		\hline
		$0$         & $0$ & $1$ \\
		\hline
		$1$         & $1$ & $1$
	\end{tabular}
	\hfill
	\begin{tabular}{c||c|c|c}
		$f_\meta^{\gor}$ & $0$ & $1$ & $\meta$ \\
		\hline
		\hline
		$0$     & $0$ &     $1$ & $\meta$ \\
		\hline
		$1$     & $1$ &     $1$ &     $1$ \\
		\hline
		$\meta$ & $\meta$ & $1$ & $\meta$
	\end{tabular}
	\hfill{}
	\caption{Gate behavior under metastability corresponds to Kleene's $3$-valued logic.}
	\label{tab:gate}
\end{table}
%}\medskip

As an example, consider Table~\ref{tab:gate} and the \gand-gate with two input ports implementing $f^{\gand}(x_1, x_2) = x_1 \land x_2$.
We extend $f^{\gand}\colon \B^2 \to \B$ to $f^{\gand}_\meta\colon \BM^2 \to \BM$.
For $x \in \B^2$, we have $f^{\gand}(x) = f_\meta^{\gand}(x)$.
Now consider $x = \meta1$.
We have $\Res(\meta1) = \{ 01, 11 \}$, so $O = \{ f^{\gand}(01), f^{\gand}(11) \} = \{0,1\} = \B$, and thus $f_\meta^{\gand}(\meta1) = \meta$.
For $x = \meta0$ we obtain $\Res(x) = \{ 00, 10 \}$, and $O = \{ f^{\gand}(00), f^{\gand}(10) \} = \{0\}$.
Hence, $f_\meta^{\gand}(\meta0) = 0$, i.e., the metastable bit is masked.

The \gor-gate is handled analogously.
Refer to Figure~\ref{fig:mux-metapropagation} for an example of metastability propagation through combinational logic.

\subsection{Combinational Logic}
\label{sec:model-combinational}

\begin{figure}
	\begin{center}
		\def\svgwidth{.4\columnwidth}
		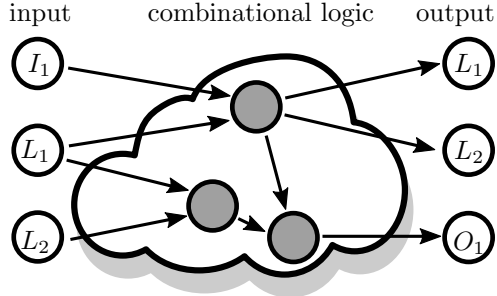
	\end{center}
	\caption{%
		Combinational logic \acs{DAG} with gates (gray) and registers (white).
		The input ($I_1$), output ($O_1$), and local ($L_1$ and $L_2$) registers occur as input nodes, output nodes, and both, respectively.}
	\label{fig:dag}
\end{figure}

We model combinational logic as \ac{DAG} $G = (V,A)$ with parallel arcs, compare Figure~\ref{fig:dag}.
Each node either is an \emph{input node,} an \emph{output node,} or a gate (see Section~\ref{sec:model-gates}).

Input nodes are sources in the \ac{DAG}, i.e., have indegree $0$ and an arbitrary outdegree, and output nodes are sinks with indegree~$1$, i.e., have indegree $1$ and outdegree~$0$.
If $v \in V$ is a gate, denote by $f_v\colon \BM^{k_v} \to \BM$ its gate function with $k_v \in \N_0$ parameters.
For each parameter of~$f_v$, $v$~is connected to exactly one input node or gate $w$ by an arc $(w,v) \in A$.
Every output node $v$ is connected to exactly one input node or gate $w$ by an arc $(w,v) \in A$.
Note that input nodes and gates can serve as input to multiple gates and output nodes.
%W.l.o.g., we require that gate nodes have an outdegree of at least one, i.e., that their outputs are used somewhere. % not needed anywhere as of 2016-04-01

Suppose $G$ has $m$ input nodes and $n$ output nodes.
Then $G$ defines a function $f^G\colon \BM^m \to \BM^n$ as follows.
Starting with input $x \in \BM^m$, we evaluate the nodes $v \in V$.
If $v$ is an input node, it evaluates to~$x_v$.
Gates of indegree $0$ are constants and evaluate accordingly.
If $v$ is a gate of non-zero indegree, it evaluates to~$f_v(\bar{x})$, where $\bar{x} \in \BM^{k_v}$ is the recursive evaluation of all nodes $w$ with $(w,v) \in A$.
Otherwise, $v$~is an output node, has indegree~$1$, and evaluates just as the unique node $w$ with $(w,v) \in A$.
Finally, $f^G(x)_v$ is the evaluation of the output node~$v$.

\subsection{Circuits}
\label{sec:model-circuits}

We formally define a circuit in this section, specify how it behaves in Section~\ref{sec:model-executions}, and give an example in Section~\ref{sec:model-example}.

\begin{definition}[Circuit]\label{def:circuit}
	A \emph{circuit $C$} is defined by:
	\begin{enumerate}
	\item
		$m$ \emph{input registers,} $k$ \emph{local registers,} and $n$ \emph{output registers,} $m,k,n \in \N_0$.
		Each register has exactly one \emph{type}\dash---simple, mask-$0$, or mask-$1$ (see Section~\ref{sec:model-registers})\dash---and is either input, output, or local register.

	\item
		A combinational logic \ac{DAG} $G$ as defined in Section~\ref{sec:model-combinational}.
		$G$~has $m + k$ input nodes, exactly one for each non-output register, and $k + n$ output nodes, exactly one for each non-input register.
		Local registers appear as both input node and output node.

	\item
		An initialization $x_0 \in \BM^{k+n}$ of the non-input registers.
	\end{enumerate}
	Each $s \in \BM^{m+k+n}$ defines a \emph{state} of~$C$.
\end{definition}
A meaningful application clearly uses a stable initialization $x_0 \in \B^{k+n}$;
this restriction, however, is not formally required.
Furthermore, observe that Definition~\ref{def:circuit} does not allow registers to be an input and an output register at the same time.
This overlap in responsibilities, however, is often used in digital circuits.
We note that we impose this restriction for purely technical reasons;
our model supports registers that are read and written\dash---local registers\dash---and it is possible to emulate the abovementioned behavior.\footnote{%
	Copy the input into local register in the first round.
	Then use the local register in the role where it is both read and written in every round.
	If needed, copy the content of the local register to an output register in every round.
}
Hence, this formal restriction has no practical implications.

We denote by
\begin{align}
	\In\colon  & \BM^{m+k+n} \to \BM^m, \\
	\Loc\colon & \BM^{m+k+n} \to \BM^k\text{, and} \\
	\Out\colon & \BM^{m+k+n} \to \BM^n
\end{align}
the projections of a circuit state to its values at input, local, and output registers, respectively.
In fact, the initialization of the output registers, $\Out(x_0)$, is irrelevant, because output registers are never read (see below).
We use the convention that for any state~$s$, $s = \In(s) \circ \Loc(s) \circ \Out(s)$.

\subsection{Executions}
\label{sec:model-executions}

Consider a circuit $C$ in state~$s$, and let $x = \In(s) \circ \Loc(s)$ be the state of the non-output registers.
Suppose each register $R$ is read, i.e., makes a non-dashed state transition according to its type, state, and corresponding state machine in Figure~\ref{fig:register-states}.
This state transition yields a value read from, as well as a new state for,~$R$.
We denote by
\begin{equation}\label{eq:read}
	\Read^C\colon \BM^{m+k} \to \Pow\left( \BM^{m+k} \right)
\end{equation}
the function mapping $x$ to the set of possible values read from non-output registers of $C$ depending on~$x$.
When only simple registers are involved, the read operation is deterministic:
\begin{observation}\label{obs:simple-read}
	In a circuit $C$ with only simple registers, $\Read^C(x) = \{ x \}$.
\end{observation}

\begin{proof}
	By Figure~\ref{fig:register-states-simple}, the only non-dashed state transition for simple registers in state $x \in \BM$ has output~$x$.
\end{proof}

In the presence of masking registers, $x \in \Read^C(x)$ can occur, but the output may partially stabilize:
\begin{observation}\label{obs:read}
	Consider a circuit $C$ in state~$s$.
	Then for $x = \In(s) \circ \Loc(s)$
	\begin{gather}
		x \in \Read^C(x)\text{, and} \label{eq:read-id}\\
		\Read^C(x) \subseteq \ResM(x). \label{eq:read-resolve}
	\end{gather}
\end{observation}

\begin{proof}
	Check the non-dashed state transitions in Figure~\ref{fig:register-states}.
	For~\eqref{eq:read-id}, observe that in all state machines, a state transition with output $b \in \BM$ starts in state~$b$.
	Regarding~\eqref{eq:read-resolve}, observe that registers in state~\meta are not restricted by the claim, and registers of any type in state $b \in \B$ are deterministically read as $b \in \ResM(b) = \{b\}$.
\end{proof}

Let $G$ be the combinational logic \ac{DAG} of $C$ with $m + k$ input and $k + n$ output nodes.
Suppose $o \in \BM^{m+k}$ is read from the non-output registers.
Then the combinational logic of $C$ evaluates to $f^G(o)$, uniquely determined by $G$ and~$o$.
We denote all possible evaluations of $C$ w.r.t.\ $x$ by $\Eval^C(x)$:
\begin{gather}
	\Eval^C\colon \BM^{m+k} \to \Pow\left( \BM^{k+n} \right), \\
	\Eval^C(x) := \left\{ f^G(o) \mid o \in \Read^C(x) \right\}.
\end{gather}

When registers are written, we allow, but do not require, signals to stabilize.
If the combinational logic evaluates the new values for the non-input registers to $\bar{x} \in \BM^{k+n}$, their new state is in $\ResM(\bar{x})$;
the input registers are never overwritten.
We denote this by
\begin{gather}
	\Write^C\colon \BM^{m+k} \to \Pow\left( \BM^{k+n} \right), \\
	\Write^C(x) := \bigcup_{\bar{x} \in \Eval^C(x)} \ResM(\bar{x}).
\end{gather}
Observe that this is where metastability can cause inconsistencies:
If a gate is read as \meta and this is copied to three registers, it is possible that one stabilizes to~$0$, one to~$1$, and one remains~\meta.

For the sake of presentation, we write $\Read^C(s)$, $\Eval^C(s)$, and $\Write^C(s)$ for a circuit state $s \in \BM^{m+k+n}$, meaning that the irrelevant part of $s$ is ignored.

Let $s_r$ be a state of~$C$.
A \emph{successor state $s_{r+1}$ of $s_r$} is any state that can be obtained from $s_r$ as follows.
\begin{description}
\item [Read phase]
	First read all registers, resulting in read values $o \in \Read^C(s_r)$.
	Let $\iota_{r+1} \in \BM^m$ be the state of the input registers after the state transitions leading to reading~$o$.

\item [Evaluation phase]
	Then evaluate the combinational logic according to the result of the read phase to $\bar{x}_{r+1} = f^G(o) \in \Eval^C(s_r)$.

\item [Write phase]
	Pick a partial resolution $x_{r+1} \in \ResM(\bar{x}_{r+1}) \subseteq \Write^C(s_r)$ of the result of the evaluation phase.
	The successor state is $s_{r+1} = \iota_{r+1} \circ x_{r+1}$.
\end{description}
In each clock cycle, our model determines some successor state of the current state of the circuit;
we refer to this as \emph{round}.

Note that due to worst-case propagation of metastability, the evaluation phase is deterministic, while read and write phase are not:
Non-determinism in the read phase is required to model the non-deterministic read behavior of masking registers, and non-determinism in the write phase allows copies of metastable bits to stabilize inconsistently.
In a physical circuit, metastability may resolve within the combinational logic;
we do not model this as a non-deterministic evaluation phase, however, as it is equivalent to postpone possible stabilization to the write phase.

Let $C$ be a circuit in state~$s_0$.
For $r \in \N_0$, an \emph{$r$-round execution (w.r.t.~$s_0$) of $C$} is a sequence of successor states $s_0, s_1, \dots, s_r$.
We denote by $S^C_r(s_0)$ the set of possible states resulting from $r$-round executions w.r.t.\ $s_0$ of~$C$:
\begin{align}
	S^C_0(s_0) &:= \{ s_0 \}\text{, and} \\
	S^C_r(s_0) &:= \left\{ s_r \mid \text{$s_r$ successor state of some $s \in S^C_{r-1}(s_0)$} \right\}.
\end{align}
An \emph{initial state of $C$ w.r.t.\ input $\iota \in \BM^m$} is $s_0 = \iota \circ x_0$.
We use $C_r\colon \BM^m \to \Pow(\BM^n)$ as a function mapping an input to all possible outputs resulting from $r$-round executions of~$C$:
\begin{equation}\label{eq:C}
	C_r(\iota) := \left\{ \Out(s_r) \mid s_r \in S^C_r(\iota \circ x_0) \right\}.
\end{equation}
We say that \emph{$r$ rounds of $C$ implement $f\colon \BM^m \to \Pow(\BM^n)$} if and only if $C_r(\iota) \subseteq f(\iota)$ for all $\iota \in \BM^m$, i.e., if all $r$-round executions of $C$ result in an output permitted by~$f$.
If there is some $r \in \N$, such that $r$ rounds of $C$ implement~$f$, we say that $C$ implements~$f$.

Observe that our model behaves exactly like a traditional, deterministic, binary circuit model if $s_0 \in \B^{m+k+n}$.

\subsection{Example}
\label{sec:model-example}

\begin{figure}
\hfill\subfigure[Circuit]{
	\begin{tikzpicture}[scale=1,transform shape,cktbaselength=0.5pt]
		\draw (0,0.5) node[or2,draw,rotate=90] (or) {};
		\draw (1.5,0) node[and2,draw,rotate=90] (and) {};

		\draw (or.b) -- ++(-0.6,0) node[anchor=east, left] (i1) {$I_1$};
		\draw (or.a) -- (i1.east |- or.a) node[anchor=east, left] (i2) {$I_2$};
		\draw (and.a) -- (i2.east |- and.a) node[anchor=east, left] (l1in) {$L_1$};
		\draw (or.z) -- ++(2,0) node[anchor=west, right] (l1out) {$L_1$};

		\coordinate (split) at ($ (or.z) + (0.4,0) $);
		\draw (split) node[circle, fill=black, inner sep=0, minimum size=3pt] {}
			-- (split |- and.b)
			-- (and.b);
		\draw (and.z) -- (l1out.west |- and.z)
			node[anchor=west, right] (o1) {$O_1$};
	\end{tikzpicture}
\label{fig:model-example-circuit}
}\hfill\subfigure[States, reads, evaluations, and writes]{
	\begin{tabular}{r|cccc|ccc|cc|cc}
		\multirow{2}{*}{$r$}
			& \multicolumn{4}{|c}{state $s_r$}
			& \multicolumn{3}{|c}{read $o$}
			& \multicolumn{2}{|c}{eval $\bar{x}_{r+1}$}
			& \multicolumn{2}{|c}{write $x_{r+1}$} \\

			& $I_1$ & $I_2$ & $L_1$ & $O_1$
			& $I_1$ & $I_2$ & $L_1$
			& $L_1$ & $O_1$
			& $L_1$ & $O_1$ \\
		\hline
		$0$
			& $\meta$ & $\meta$ & $1$ & $1$
			& $0$ & $\meta$ & $1$
			& $\meta$ & $\meta$
			& $1$ & $\meta$ \\
		$1$
			& $\meta$ & $\meta$ & $1$ & $\meta$
			& $\meta$ & $\meta$ & $1$
			& $\meta$ & $\meta$
			& $\meta$ & $\meta$ \\
		$2$
			& $1$ & $\meta$ & $\meta$ & $\meta$
			& $1$ & $\meta$ & $\meta$
			& $1$ & $\meta$
			& $1$ & $0$ \\
		$3$
			& $1$ & $\meta$ & $1$ & $0$
			& $1$ & $\meta$ & $1$
			& $1$ & $1$
			& $1$ & $1$ \\
		$4$
			& $1$ & $\meta$ & $1$ & $1$
			& & &
			& &
			& &
	\end{tabular}
\label{fig:model-example-states}
}\hfill{}
\caption{%
	Example execution in a circuit~\subref{fig:model-example-circuit}.
	The node states as well as the results of the read, evaluation, and write phases are listed in the table~\subref{fig:model-example-states}.
	Register $I_1$ is a mask-$0$ register, all others are simple registers.
	The initialization is~$11$, the input is~$\meta\meta$, and hence $s_0 = \meta\meta11$.
}
\label{fig:model-example}
\end{figure}
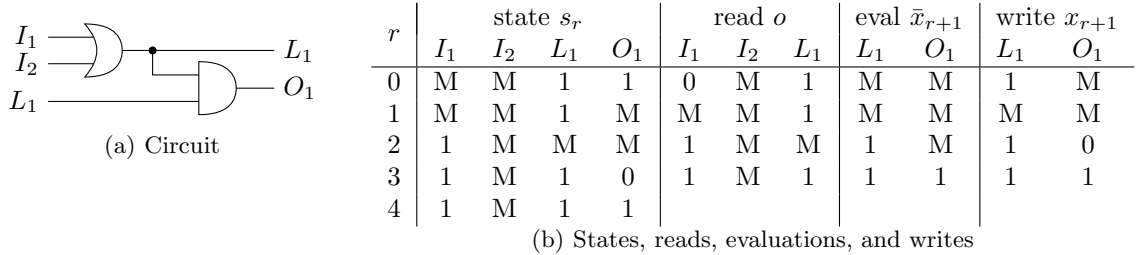

We use this section to present an example of our model.
Figure~\ref{fig:model-example} specifies a circuit and its states, as well as the results of the read, evaluation, and write phases.
The input registers are $I_1$ and~$I_2$, the only local register is~$L_1$, and the only output register is~$O_1$.
Regarding register types, the input register $I_1$ is a mask-$0$ register and all other registers are simple registers.

The initialization is $x_0 = 11$, the input is $\iota = \meta\meta$, and the initial state hence is $s_0 = \iota \circ x_0 = \meta\meta11$, which is indicated in the upper left entry in Figure~\ref{fig:model-example-states}.
In the read phase, all non-output registers are read.
Since $I_2$ and $L_1$ are simple registers, their read deterministically evaluates to $\meta$ and~$1$, respectively, by the state machine in Figure~\ref{fig:register-states-simple}.
The mask-$0$ register $I_1$ in state $\meta$ may either be read as $0$ and remain in state~$\meta$, or be read as $\meta$ and transition to state~$1$, compare Figure~\ref{fig:register-states-mask0};
in this case it does the former.
So far, we fixed the outcome of the read phase, $0\meta1$, and the follow-up state of the input registers,~$\meta\meta$;
the other registers are overwritten at the end of the write phase.
The evaluation is uniquely determined, a read phase resulting in $o$ evaluates to $f^G(o)$, here, $f^G(0\meta1) = \meta\meta$.
We are left with only one more step in this round:
The non-input registers are overwritten with some value in the resolution of the evaluation phase's result, in our case with $1\meta \in \ResM(\meta\meta)$.
Together we obtain the successor state $s_1 = \meta\meta1\meta$.

In the next round, $I_1$~uses the other state transition, i.e., is read as~$\meta$, and hence has state $1$ in the next round.
Hence its state remains fixed in all successive rounds by the state machine in Figure~\ref{fig:register-states-mask0}.
The other reads are deterministic, so we obtain $o = \meta\meta1$ as the result of the read phase and successor states $1\meta$ for $I_1$ and~$I_2$.
The evaluation is $f^G(o) = f^G(\meta\meta1) = \meta\meta$ the state of $L_1$ and $O_1$ is overwritten with some value from $\ResM(\meta\meta)$, here by~$\meta\meta$.

By round $r = 2$, the result of the read phase is deterministic because the only masking register stabilized, we read $o = 1\meta\meta$, and evaluate to~$1\meta$.
The remaining non-determinism is whether to write $1\meta$ or some stabilization thereof. We examine the case that $10$ is written.

Rounds $r \geq 3$ now are entirely deterministic.
The only possible read is $1\meta1$, which evaluates to $f^G(1\meta1) = 11$, fixing the result of the write phase to~$11$.
Further rounds are identical, the only metastable register, $I_2$, remains metastable but has no impact on the evaluation phase as the \gor gate always receives input $1$ from $I_2$ and hence masks the metastable input.

\section{Case Study: \acl{CMUX}}
\label{sec:mux}

In this section, we demonstrate the model proposed in Section~\ref{sec:model} by developing \iac{CMUX}.
Despite its simplicity, it demonstrates our concept, and is a crucial part of the more complex metastability-containing components required for the clock synchronization circuit outlined in Section~\ref{sec:arithmetic-application}~\cite{blm-nomcsn-17,fklp-mametdc-17,lm-emcgc2s-16}.
From a broader perspective, this section shows that our model, especially the worst-case propagation of metastability, is not ``too pessimistic'' to permit positive results.
We show in Section~\ref{sec:realitycheck} that it is not ``too optimistic,'' either.

\begin{figure*}{\small % match caption size
\hfill\subfigure[$C^{\text{\acs{MUX}1}}$]{
	\begin{tikzpicture}[scale=1,transform shape,cktbaselength=0.5pt]
		\draw (0,0) node[or2,draw,rotate=90] (o1) {};
		\draw (o1.a)
			-- ++(-0.3,0)
			-- ++(0,-0.2)
			-- ++(-0.3,0) node[and2,draw,rotate=90,anchor=south] (a1) {};
		\draw (o1.b)
			-- ++(-0.3,0)
			-- ++(0,+0.2)
			-- ++(-0.3,0) node[and2ni,draw,rotate=90,anchor=south] (a2) {};

		\draw (a1.a) -- ++(-0.6,0) node[xshift=-5pt] {$b$};
		\draw (a2.a) -- ++(-0.6,0) node[xshift=-5pt] {$a$};

		\draw (a2.b)
			-- ++(-0.3,0)
			-- ++(0,-1.6) node[yshift=-5pt] (s) {$s$};
		\draw (a1.b) -| (s);
		\path (a1.b) -- ++(-0.4,0) node[circle,fill=black,inner sep=0,minimum size=3pt] {};

		\draw (o1.z) -- ++(0.3,0) node[xshift=5pt] {$o$};
	\end{tikzpicture}
\label{fig:mux_gate}
}\hfill
\subfigure[$C^{\text{\acs{MUX}2}}$]{
	\begin{circuitikz}[scale = 0.8, transform shape]
		\draw (0,0) node[nmos,rotate=90] (mos1) {};
		\draw (0,0) node[pmos,rotate=-90] (mos1n) {};
		\draw (mos1.D) node[left=0.01] {$b$};

		\draw (0,-2) node[nmos,rotate=90] (mos2) {};
		\draw (0,-2) node[pmos,rotate=-90] (mos2n) {};
		\draw (mos2.D) node[left=0.01] {$a$};

		\draw (mos1.G) to (mos2n.G);
		\draw (mos1.S) to (mos2.S) to [short,*-] ++(0.6,0) {} node[right=0.05] {$o$};

		\draw (-1.8,-0.5) node[not port, scale=0.7] (not1) {};
		\draw (not1.in)
			to ++(-0.2,0)
			to[short,-*] ++(0,-0.5)
			to ++(-0.2,0) node[left=0.05] {$s$}
			to[short,-*] ++(2.69,0);
		\draw (not1.out)
			to ++(0.2,0)
			to ++(0,0.1) |- (mos1n.G);
		\draw (not1.out)
			to [short,-*] ++(0.2,0)
			to ++(0,-0.1) |- (mos2.G);
	\end{circuitikz}
\label{fig:mux_trans}
}\hfill
\subfigure[$C^{\text{\acs{CMUX}1}}$]{
	\begin{tikzpicture}[scale=1,transform shape,cktbaselength=0.5pt]
		\draw (0,0) node[or3,draw,rotate=90] (o1) {};
		\draw (o1.a)
			-- ++(-0.3,0)
			-- ++(0,-0.5)
			-- ++(-0.3,0) node[and2,draw,rotate=90,anchor=south] (a1) {};
		\draw (o1.b) -- ++(-0.6,0) node[and2,draw,rotate=90,anchor=south] (am) {};
		\draw (o1.c)
			-- ++(-0.3,0)
			-- ++(0,+0.5)
			-- ++(-0.3,0) node[and2ni,draw,rotate=90,anchor=south] (a2) {};

		\draw (a1.a) -- ++(-0.6,0) node[xshift=-5pt] {$b$};
		\draw (a2.a) -- ++(-0.6,0) node[xshift=-5pt] {$a$};

		\draw (a2.b)
			-- ++(-0.3,0)
			-- ++(0,-2.3) node[yshift=-5pt] (s) {$s$};
		\draw (a1.b) -| (s);
		\path (a1.b) -- ++(-0.4,0) node[circle,fill=black,inner sep=0,minimum size=3pt] {};

		\draw (am.b) to ++(-0.2,0) to ++(0,0.5)  node[circle,fill=black,inner sep=0,minimum size=3pt] {};
		\draw (am.a) to ++(-0.2,0) to ++(0,-0.85)  node[circle,fill=black,inner sep=0,minimum size=3pt] {};

		\draw (o1.z) -- ++(0.3,0) node[xshift=5pt] {$o$};
	\end{tikzpicture}
\label{fig:fixedmux}
}\hfill
\subfigure[$C^\text{\acs{CMUX}2}$]{
	\begin{tikzpicture}[scale=1,transform shape,cktbaselength=0.5pt]
		\draw (0,0) node[or2,draw,rotate=90] (o1) {};
		\draw (o1.a)
			-- ++(-0.3,0)
			-- ++(0,-0.2)
			-- ++(-0.3,0) node[and2,draw,rotate=90,anchor=south] (a1) {};
		\draw (o1.b)
			-- ++(-0.3,0)
			-- ++(0,+0.2)
			-- ++(-0.3,0) node[and2ni,draw,rotate=90,anchor=south] (a2) {};

		\draw (a1.a) -- ++(-0.5,0) node[xshift=-5pt] {$b$};
		\draw (a2.a) -- ++(-0.5,0) node[xshift=-5pt] {$a$};

		\path (a1.b) -- ++(-0.5,0) node[rounded corners=3pt, minimum width=0.7cm,text height=0.15cm,draw] (delay) {$\Delta$};
		\draw (a1.b) -- (delay.east);

		\draw (delay.west)
			-- ++(-0.1,0)
			-- ++(0,-1.0) node[yshift=-5pt] (s) {$s$};
		\draw (a2.b) -| (s);
		\node[xshift=24pt,yshift=0pt] at (s) {(mask-$1$)};
		\path (a1.b) -- ++(-0.955,0) node[circle,fill=black,inner sep=0,minimum size=3pt] {};

		\draw (o1.z) -- ++(0.3,0) node[xshift=5pt] {$o$};
	\end{tikzpicture}
\label{fig:mux_timed}
}\hfill}
\caption{%
	\ac{MUX} implementations.
	Figures~\subref{fig:mux_gate} and~\subref{fig:mux_trans} depict the gate-level circuit and the transmission gate implementation of a standard \ac{MUX}.
	The circuits in Figures~\subref{fig:fixedmux} and~\subref{fig:mux_timed} mask a metastable select bit $s$ in the case of $a = b$ employing additional gates~\subref{fig:fixedmux} and a masking register~\subref{fig:mux_timed}, respectively.%
}
\label{fig:mux_all}
\end{figure*}
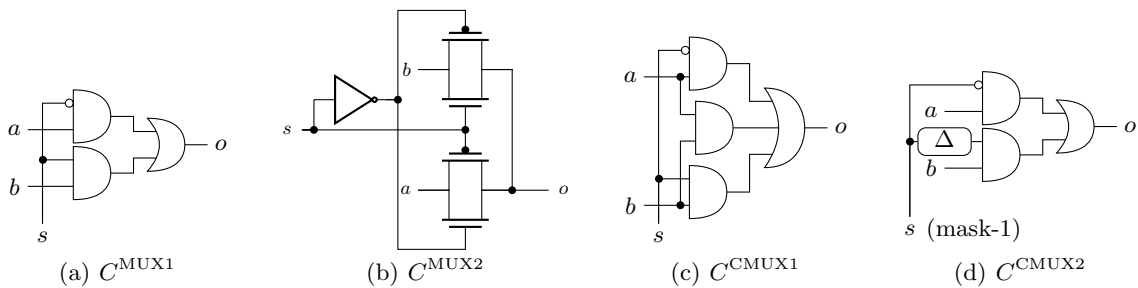

Prior to discussing improved variants, let us examine a standard \ac{MUX}.
A \emph{($k$-bit) \acf{MUX}} is a circuit $C$ with $2k+1$ inputs, such that $C$ implements
\begin{align}
	f_{\text{\acs{MUX}}}\colon &\BM^k \times \BM^k \times \BM \to \BM^k \\
	f_{\text{\acs{MUX}}}(a,b,s) &= \begin{cases}
		\ResM(a) & \text{if $s = 0$,} \\
		\ResM(b) & \text{if $s = 1$, and} \\
		\BM      & \text{if $s = \meta$,}
	\end{cases}\label{eq:mux}
\end{align}
where we use $k = 1$ for the sake of presentation.
In the case of a stable select bit~$s$, it determines whether to output (some stabilization of) $a$ or~$b$.
If $s$ is metastable, an arbitrary output may be produced.
Figures~\ref{fig:mux_gate} and~\ref{fig:mux_trans} show typical implementations in terms of combinational logic and transmission gates, respectively.

\begin{figure}{\small % match caption size
\hfill\subfigure[$C^{\text{\acs{MUX}1}}$]{
	\begin{tikzpicture}[scale=1,transform shape,cktbaselength=0.5pt]
		\fill [fill=black!10] (-1.825, -1.05) rectangle (0.7, 0.9);
		\node[draw=none,fill=black!10] at (-0.78, 1.25) {metastability};

		\draw (0,0) node[or2,draw,rotate=90] (o1) {};
		\draw (o1.a)
			-- ++(-0.3,0)
			-- ++(0,-0.2)
			-- ++(-0.3,0) node[and2,draw,rotate=90,anchor=south] (a1) {};
		\draw (o1.b)
			-- ++(-0.3,0)
			-- ++(0,+0.2)
			-- ++(-0.3,0) node[and2ni,draw,rotate=90,anchor=south] (a2) {};

		\draw (a1.a) -- ++(-0.6,0) node[anchor=east, left] {$b=1$};
		\draw (a2.a) -- ++(-0.6,0) node[anchor=east, left] {$a=1$};

		\draw (a1.a) node[anchor=east, below left, xshift=1pt] {$1$};
		\draw (a1.b) node[anchor=east, below left, xshift=1pt] {$\meta$};
		\draw (a1.z) node[below right] {$\meta$};
		\draw (a2.a) node[anchor=east, above left, xshift=1pt] {$1$};
		\draw (a2.b) node[anchor=east, above left, xshift=4pt] {$\meta$};
		\draw (a2.z) node[above right] {$\meta$};

		\draw (a2.b)
			-- ++(-0.3,0)
			-- ++(0,-1.6) node[below] (s) {$s=\meta$~~~~~~~~};
		\draw (a1.b) -| (s);
		\path (a1.b) -- ++(-0.4,0) node[circle,fill=black,inner sep=0,minimum size=3pt] {};

		\draw (o1.z) -- ++(0.3,0) node[below, xshift=3pt] {$o=\meta$};
	\end{tikzpicture}
\label{fig:mux-metapropagation-1}
}\hfill
\subfigure[$C^{\text{\acs{CMUX}1}}$]{
	\begin{tikzpicture}[scale=1,transform shape,cktbaselength=0.5pt]
		\fill [fill=black!10] (-1.875,-1.4) rectangle (-0.35,-0.4);
		\fill [fill=black!10] (-1.875, 1.4) rectangle (-0.35, 0.4);
		\fill [fill=black!10] (-1.875,-1.4) rectangle (-1.6, 1.325);

		\draw (0,0) node[or3,draw,rotate=90] (o1) {};
		\draw (o1.a)
			-- ++(-0.3,0)
			-- ++(0,-0.5)
			-- ++(-0.3,0) node[and2,draw,rotate=90,anchor=south] (a1) {};
		\draw (o1.b) -- ++(-0.6,0) node[and2,draw,rotate=90,anchor=south] (am) {};
		\draw (o1.c)
			-- ++(-0.3,0)
			-- ++(0,+0.5)
			-- ++(-0.3,0) node[and2ni,draw,rotate=90,anchor=south] (a2) {};

		\draw (a1.a) -- ++(-0.6,0) node[anchor=east, left] {$b=1$};
		\draw (a2.a) -- ++(-0.6,0) node[anchor=east, left] {$a=1$};

		\draw (a2.b)
			-- ++(-0.3,0)
			-- ++(0,-2.3) node[below] (s) {$s=\meta$~~~~~~~~};
		\draw (a1.b) -| (s);
		\path (a1.b) -- ++(-0.4,0) node[circle,fill=black,inner sep=0,minimum size=3pt] {};

		\draw (am.b) to ++(-0.2,0) to ++(0,0.5)  node[circle,fill=black,inner sep=0,minimum size=3pt] {};
		\draw (am.a) to ++(-0.2,0) to ++(0,-0.85)  node[circle,fill=black,inner sep=0,minimum size=3pt] {};

		\draw (a1.a) node[anchor=east, below left] {$1$};
		\draw (a1.b) node[anchor=east, below left] {$\meta$};
		\draw (a1.z) node[below right] {$\meta$};

		\draw (a2.a) node[anchor=east, above left] {$1$};
		\draw (a2.b) node[anchor=east, above left, xshift=3pt] {$\meta$};
		\draw (a2.z) node[above right] {$\meta$};

		\draw (am.a) node[anchor=east, left, xshift=-3pt] {$1$};
		\draw (am.b) node[anchor=east, left, xshift=-3pt] {$1$};
		\draw (am.z) node[above right] {$1$};

		\draw (o1.z) -- ++(0.3,0) node[below, xshift=1pt] {$o=1$};
	\end{tikzpicture}
\label{fig:mux-metapropagation-2}
}\hfill
}
\caption{%
	\ac{MUX} behavior for $a = b = 1$, in which case the output should be~$1$, regardless of the select bit~$s$.
	For $s = \meta$, however, the standard \ac{MUX}~\subref{fig:mux-metapropagation-1} can become metastable, but the \ac{CMUX}~\subref{fig:mux-metapropagation-2} outputs~$1$.%
}
\label{fig:mux-metapropagation}
\end{figure}
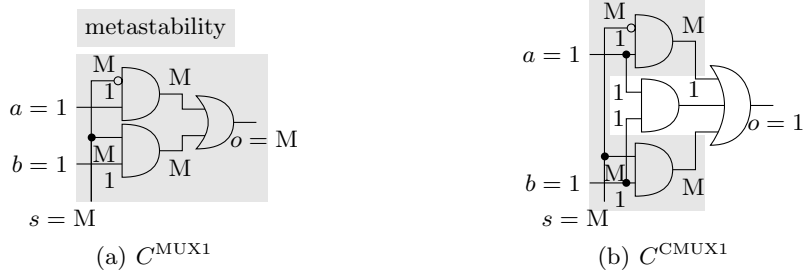

A desirable property of \iac{MUX} is that if $a = b$, the output is~$a$, regardless of~$s$.
Being uncertain whether to select~$a$ or~$b$ should be insubstantial in this case.
If, however, $s = \meta$ and $a = b = 1$, a standard implementation like $C^{\text{\acs{MUX}1}}$, compare Figure~\ref{fig:mux-metapropagation-1}, yields
\begin{equation}
	(\lnot s \land a) \lor (s \land b)
		= (\lnot \meta \land 1) \lor (\meta \land 1)
		= \meta \lor \meta
		= \meta.
\end{equation}
Hence, we ask for an improved circuit that implements
\begin{align}
	f_{\text{\acs{CMUX}}}\colon & \BM^k \times \BM^k \times \BM \to \BM^k \\
	f_{\text{\acs{CMUX}}}(a,b,s) &= \begin{cases}
			\ResM(a) & \text{if $s = 0$ or $a = b$,} \\
			\ResM(b) & \text{if $s = 1$, and} \\
			\BM      & \text{if $a \neq b \land s = \meta$.}
	\end{cases}\label{eq:mux.M}
\end{align}
We call such a circuit \emph{($k$-bit) \acf{CMUX}}.
Circuit~$C^{\text{\acs{CMUX}1}}$ in Figure~\ref{fig:fixedmux} implements~\eqref{eq:mux.M}:
The problematic case of $s = \meta$ and $a = b = 1$ is handled by the third \gand-gate which becomes~$1$, providing the \gor-gate with a stable $1$ as input, see Figure~\ref{fig:mux-metapropagation-2}.

\begin{lemma}
	$C^{\text{\acs{CMUX}1}}_1 \subseteq f_{\text{\acs{CMUX}}}$ from Equation~\eqref{eq:mux.M}.
\end{lemma}

\begin{proof}
	$C^{\text{\acs{CMUX}1}}$ has no internal registers and its combinational logic \ac{DAG} implements
	\begin{equation}\label{eq:mmux1proof}
		o = (\lnot s \land a) \lor (s \land b) \lor (a \land b).
	\end{equation}
	It is easy to check that for $s \neq \meta$, \eqref{eq:mmux1proof} implements the first two cases of~\eqref{eq:mux.M}.
	If $s = \meta$, and $a \neq b$ or $a = b = \meta$, $C^{\text{\acs{CMUX}1}}$ may output anything, so consider $s = \meta$ and distinguish two cases:
	\begin{enumerate*}
	\item
		If $a = b = 0$, all clauses in~\eqref{eq:mmux1proof} are~$0$, hence $o = 0$, and

	\item
		if $a = b = 1$, $a \land b = 1$ and $o = 1$, regardless of the other clauses.
	\end{enumerate*}
\end{proof}

The price for this improvement is an additional \gand-gate and a ternary \gor-gate, which can be costly if $a$ and $b$ are of large bit width.
We reduce the gate number using a masking register to implement~\eqref{eq:mux.M} in two steps.
First, we show how to implement~\eqref{eq:mux.M} using two rounds in our model, and then derive from it an efficient unclocked physical implementation with fewer gates (this unclocked implementation is not covered by our model, see below).
Algorithm~\ref{alg:mmux} specifies the clocked circuit by assignments of logic expressions to registers.
The trick is to sequentially read $s$ from a mask-$1$ register, ensuring that at most one copy of $s$ can be metastable, compare Figure~\ref{fig:register-states-mask1}.
This guarantees that in the case of $s = \meta$ and $a = b = 1$, one of the \gand-clauses is stable~$1$.

\begin{algorithm}%[ht]
	\begin{algorithmic}
		\algblockdefx{Round}{End}{{\bf each round:}}{{\bf end}}
		\State {\bf input:} $a$ and $b$ (simple), $s$ (mask-$1$)
		\State {\bf local:} $s'$ (simple)
		\State {\bf output:} $o$ (simple)
		\Statex
		\Round
			\State $s' \gets s$
			\State $o \gets (\lnot s \land a) \lor (s' \land b)$
		\End
	\end{algorithmic}
	\caption{\acl{CMUX}.}
	\label{alg:mmux}
\end{algorithm}

%The circuit has input registers $a$, $b$, and $s$, one local register~$s'$, and output register~$o$.
%The combinational logic performs in each round the following assignments:
%\begin{align}\label{eq:mmux-code}
%	s' \gets s\quad \text{and}\quad
%	o  \gets (\lnot s \land a) \lor (s' \land b)
%\end{align}

\begin{lemma}
	Two rounds of Algorithm~\ref{alg:mmux} implement~\eqref{eq:mux.M}.
\end{lemma}

\begin{proof}
	If $s \neq \meta$, we have $s = s'$ after round~$1$ and the first two cases of~\eqref{eq:mux.M} are easily verified.
	In case $s = \meta$, and $a \neq b$ or $a = b = \meta$, the output is not restricted.
	Hence, consider $s = \meta$ and $a = b$.
	If $s = \meta$ and $a = b = 0$, $o = (\lnot s \land 0) \lor (s' \land 0) = 0$.
	If $s = \meta$ and $a = b = 1$, the read and write phases of round $1$ have two possible outcomes (compare Figure~\ref{fig:register-states-mask1}):
	\begin{enumerate*}
	\item
		$s$ is read as~\meta, so its copy in $s'$ may become metastable, but $s$ is guaranteed to be read as $0$ in round~$2$ because $s$ is a mask-$1$ register.
		Then we have $o = (\lnot 0 \land 1) \lor (s' \land 1) = 1 \lor s' = 1$.

	\item
		$s' = 1$ due to $s$ masking state~\meta, in which case we obtain $o = (\lnot s \land 1) \lor (1 \land 1) = \lnot s \lor 1 = 1$.
	\end{enumerate*}
\end{proof}

One may argue that a direct realization of Algorithm~\ref{alg:mmux} in hardware as a clocked state machine may be too large for practical applications.
In fact, however, the algorithm has an optimized unclocked realization, that cannot directly be expressed in our synchronous circuit model:
The serialization of assignments in Algorithm~\ref{alg:mmux} ensured by the two clock cycles can also be enforced by local delay constraints instead of clock cycles,
see Figure~\ref{fig:mux_timed}.
With a propagation delay from $s$ to the \gand-gate with non-negated $s$ input being larger than the gate delay from $s$ to the \gand-gate with negated
  input~$\lnot s$, the circuit exhibits the specified behavior.
Note that this yields an efficient transformation of \iac{MUX} into \iac{CMUX}:
Take a standard \ac{MUX}, read the select bit from a masking register, and add the delay line.
Observe that this construction scales well with increasing bit widths of $a$ and~$b$, since only the select bit needs to be stored in a masking register.

\section{Basic Properties}
\label{sec:basics}

We establish basic properties regarding computability in the model from Section~\ref{sec:model}.
Regarding the implementability of functions by circuits, we focus on two resources:
the number $r \in \N$ of rounds and the register types available to it.
In order to capture this, let $\Fun_S^r$ be the class of functions implementable with $r$ rounds of circuits comprising only simple registers.
Analogously, $\Fun_M^r$ denotes the class of functions implementable with $r$ rounds that may use masking and simple registers.

First consider the combinational logic.
Provided with a partially metastable input~$x$, some gates\dash---those where the collective metastable input ports have an impact on the output\dash---evaluate to~\meta.
So when stabilizing $x$ bit by bit, no new metastability is introduced at the gates.
Furthermore, once a gate stabilized, its output is fixed;
stabilizing the input leads to stabilizing the output.

\begin{lemma}\label{lem:dag}
	Let $G$ be a combinational logic \ac{DAG} with $m$ input nodes.
	Then for all $x \in \BM^m$,
	\begin{equation}
		x' \in \ResM(x) \Rightarrow f^G(x') \in \ResM\left( f^G(x) \right).
	\end{equation}
\end{lemma}

\begin{proof}
	We show the statement by induction on~$|V|$.
	For the sake of the proof we extend $f^G$ to all nodes of $G = (V,A)$, i.e., write $f^G(x)_v$ for the evaluation of $v \in V$ w.r.t.\ input~$x$, regardless of whether $v$ is an output node.
	The claim is trivial for $|V| = 0$.
	Hence, suppose the claim holds for \acp{DAG} with up to $i \in \N_0$ vertices, and consider \iac{DAG} $G = (V,A)$ with $|V| = i + 1$.
	As $G$ is non-empty, it contains a sink $v \in V$.
	Removing $v$ allows applying the induction hypothesis to the remaining graph, proving that $f^G(x')_w \in \ResM(f^G(x)_w)$ for all nodes $w \neq v$.

	Concerning~$v$, the claim is immediate if $v$ is a source, because $f(x)_v = x_v$ if $v$ is an input node and $f(x)_v = b$ for a constant $b \in \BM$ if $v$ is a gate of indegree~$0$.
	If $v$ is an output node, it evaluates to the same value as the unique node $w$ with $(w,v) \in A$, which behaves as claimed by the induction hypothesis.
	Otherwise $v$ is a gate of non-zero indegree;
	consider the nodes $w \in V$ with $(w,v) \in A$.
	For input~$x$, $v$~is fed with the input string $\bar{x} \in \BM^{k_v}$, whose components are given by $f^G(x)_w$; define $\bar{x}'$ analogously w.r.t.\ input $x'$.
	Note that $\bar{x}' \in \ResM(\bar{x})$, since we already established that $f^G(x')_w \in \ResM(f^G(x))_w$ for all $w \neq v$.
	If $f_v(\bar{x}) = \meta$, the claim holds because $\ResM(\meta) = \BM$.
	On the other hand, for the case that $f_v(\bar{x}) = b \neq \meta$, our gate definition entails that $f_v(\bar{x}') = b$, because $\bar{x}' \in \ResM(\bar{x})$.
\end{proof}

Stabilizing the input of the combinational logic stabilizes its output.
The same holds for the evaluation phase:
If one result of the read phase is $x$ and another is $x' \in \ResM(x)$, the combinational logic stabilizes its output to $f^G(x') \in \ResM(f^G(x))$.
Recall Observations~\ref{obs:simple-read} and~\ref{obs:read}:
In state $x$, simple registers are deterministically read as~$x$, and masking registers as some $x' \in \ResM(x)$.
Hence, the use of masking registers might partially stabilize the input to the combinational logic and, by Lemma~\ref{lem:dag}, its output.
The same stabilization, however, can also occur in the write phase.
This implies that $\Write^C$ is not influenced by the register types.

\begin{lemma}\label{lem:1-round-writes}
	Consider a circuit~$C$ in state~$s$.
	Let $C_S$ be a copy of $C$ that only uses simple registers, and $x = \In(s) \circ \Loc(s)$ the projection of $s$ to the non-output registers.
	Then
	\begin{equation}
		\Write^C(s) = \Write^{C_S}(s) = \ResM\left( f^G(x) \right).
	\end{equation}
\end{lemma}

\begin{proof}
	In $C_S$, we have $\Read^{C_S}(s) = \{ x \}$ by Observation~\ref{obs:simple-read}.
	So $\Eval^{C_S}(s) = \{ f^G(x) \}$, and $\Write^{C_S}(s) = \ResM(f^G(x))$ by definition.

	In~$C$, $x \in \Read^C(s)$ by Observation~\ref{obs:read}, so $\ResM(f^G(x)) \subseteq \Write^C(s)$.
	All other reads $x' \in \Read^C(s)$ have $x' \in \ResM(x)$ by Observation~\ref{obs:read}, and $f^G(x') \in \ResM(f^G(x))$ by Lemma~\ref{lem:dag}.
	It follows that $\Write^C(s) = \ResM(f^G(x))$.
\end{proof}

Carefully note that the write phase only affects non-input registers; input registers are never written.
Hence, Lemma~\ref{lem:1-round-writes} does not generalize to multiple rounds:
State transitions of input registers in the read phase affect future read phases.

In $1$-round executions, however, masking and simple registers are equally powerful, because their state transitions only affect rounds $r \geq 2$ (we show in Section~\ref{sec:hierarchy-arbitrary} that these state changes lead to differences for $r \geq 2$ rounds).
\begin{corollary}\label{cor:1-round-equivalence}
	$\Fun_S^1 = \Fun_M^1$.
\end{corollary}

In contrast, simple and masking registers used as non-input registers behave identically, regardless of the number of rounds:
A circuit $C$ in state $s_r$ overwrites them regardless of their state.
Since $\Write^C(s_r)$ is oblivious to register types by Lemma~\ref{lem:1-round-writes}, so is $\Loc(s_{r+1}) \circ \Out(s_{r+1})$ for a successor state $s_{r+1}$ of~$s_r$.
\begin{corollary}\label{cor:register-equivalence}
	Simple and masking registers are interchangeable when used as non-input registers.
\end{corollary}

Consider a circuit $C$ in state~$s$, and suppose $x \in \Read^C(s)$ is read.
Since the evaluation phase is deterministic, the evaluation $y = f^G(x) \in \Eval^C(s)$ is uniquely determined by $x$ and~$C$.
Recall that we may resolve metastability to $\ResM(y) \subseteq \Write^C(s)$ in the write phase:
The state of an output register $R$ becomes $0$ if $y_R = 0$, $1$~if~$y_R = 1$, and some $b \in \BM$ if $y_R = \meta$.
Consequently, output registers resolve independently:
\begin{corollary}\label{cor:bitwise-closed}
	For any circuit~$C$, $C_1 = g_0 \times \dots \times g_{n-1}$, where $g_i\colon \BM^m \to \{ \{0\}, \{1\}, \BM \}$.
\end{corollary}

\begin{proof}
	Let $s = \iota \circ x_0$ be the initial state of $C$ w.r.t.\ input~$\iota$, and $x = \In(s) \circ \Loc(s)$.
	By Lemma~\ref{lem:1-round-writes}, $\Write^C(s)=\ResM(f^G(x))$, i.e., $C_1(\iota) = \{ \Out(s') \mid s'\in \ResM(f^G(x)) \}$.
	By definition, $\ResM(f^G(x)) = \prod_{i \in [n]} \ResM(f^G(x))_i$.
	Hence, the claim follows with $g_i(\iota) := \ResM(f^G(x))_i$ for all $\iota \in \BM^m$ and $i \in [n]$.
\end{proof}

We show in Section~\ref{sec:simple} that Corollary~\ref{cor:bitwise-closed} generalizes to multiple rounds of circuits with only simple registers.
This is, however, not the case in the presence of masking registers, as demonstrated in Section~\ref{sec:hierarchy}.

Lemmas~\ref{lem:dag} and~\ref{lem:1-round-writes} apply to the input of circuits:
Partially stabilizing an input partially stabilizes the possible inputs of the combinational logic, and hence its evaluation and the circuit's output after one round.

\begin{observation}\label{obs:specific}
	For a circuit $C$ and input $\iota \in \BM^m$,
	\begin{equation}
		\iota' \in \ResM(\iota) \Rightarrow C_1(\iota') \subseteq C_1(\iota).
	\end{equation}
\end{observation}

\begin{proof}
	Let $x_0$ be the initialization of~$C$, $s = \iota \circ x_0$ its initial state w.r.t.\ input~$\iota$, and $x = \In(s) \circ \Loc(s)$ the state of the non-output registers;
	define $s'$ and $x'$ equivalently w.r.t.\ input $\iota' \in \ResM(\iota)$.
	Using Lemmas~\ref{lem:dag} and~\ref{lem:1-round-writes}, and that $\ResM(x') \subseteq \ResM(x)$ for $x' \in \ResM(x)$, we obtain that $\Write^C(s') = \ResM(f^G(x')) \subseteq \ResM(f^G(x)) = \Write^C(s)$.
\end{proof}

Finally, note that adding rounds of computation cannot decrease computational power;
a circuit determining $x$ in $r$ rounds can be transformed into one using $r + 1$ rounds by buffering $x$ for one round.
Furthermore, allowing masking registers does not decrease computational power.
\begin{observation}\label{obs:subseteq}
	For all $r \in \N_0$ we have
	\begin{align}
		\Fun_S^r &\subseteq \Fun_S^{r+1}\text{,} \\
		\Fun_M^r &\subseteq \Fun_M^{r+1}\text{, and} \\
		\Fun_S^r &\subseteq \Fun_M^r.
	\end{align}
\end{observation}

\section{Reality Check}
\label{sec:realitycheck}

Section~\ref{sec:mux} demonstrates that our model permits the design of metastability-containing circuits.
Given the elusive nature of metastability and Marino's impossibility result~\cite{m-gtmo-81}, non-trivial positive results of this kind are surprising, and raise the question whether the proposed model is ``too optimistic'' to derive meaningful statements about the physical world.
Put frankly, a reality check is in order!

In particular, Marino established that no digital circuit can reliably
\begin{enumerate*}
\item\label{enum:marino-begin}
	avoid,

\item
	resolve, or

\item\label{enum:marino-end}
	detect
\end{enumerate*}
metastability~\cite{m-gtmo-81}.
It is imperative that these impossibility results are maintained by any model comprising metastability.
We show in Theorem~\ref{thm:pivotal} and Corollaries~\ref{cor:no-meta-detection}--\ref{cor:no-meta-resolve} that \ref{enum:marino-begin}--\ref{enum:marino-end} are impossible in the model proposed in Section~\ref{sec:model} as well.
We stress that this is about putting the model to the test rather than reproducing a known result.

We first verify that avoiding metastability is impossible in non-trivial circuits.
Consider a circuit $C$ that produces different outputs for inputs $\iota \neq \iota'$.
The idea is to observe how the output of $C$ behaves while transforming $\iota$ to $\iota'$ bit by bit, always involving intermediate metastability, i.e., switching the differing bits from~$0$ to~$\meta$ to~$1$ or vice versa.
This can be seen as a discrete version of Marino's argument for signals that map continuous time to continuous voltage~\cite{m-gtmo-81}.
Furthermore, the bit-wise transformation of $\iota$ to~$\iota'$, enforcing a change in the output in between, has parallels to the classical impossibility of consensus proof of Fischer et~al.~\cite{flp-idcofp-85};
our techniques, however, are quite different.
The following definition formalizes the step-wise manipulation of bits.

\begin{definition}[Pivotal Sequence]\label{def:pivotal}
	Let $k \in \N_0$ and $\ell \in \N$ be integers, and $x,x' \in \BM^k$.
	Then
	\begin{equation}
		\left(x^{(i)}\right)_{i \in [\ell+1]}, \quad x^{(i)} \in \BM^k,
	\end{equation}
	is a \emph{pivotal sequence (from $x$ to~$x'$ over $\BM^k$)} if and only if it satisfies
	\begin{enumerate}
	\item
		$x^{(0)} = x$ and $x^{(\ell)} = x'$,

	\item
		for all $i \in [\ell]$, $x^{(i)}$ and $x^{(i+1)}$ differ in exactly one bit, and

	\item
		this bit is metastable in either $x^{(i)}$ or~$x^{(i+1)}$.
	\end{enumerate}
	For $i \in [\ell]$, we call the differing bit the \emph{pivot from $i$ to $i + 1$} and $P_i$ its corresponding \emph{pivotal register.}
\end{definition}

Carefully note that we do not use pivotal sequences as temporal sequences of successor states;
$(x^{(i)})_{i \in [5]}$ and $(y^{(i)})_{i \in [7]}$ in Figure~\ref{fig:pivotal} do not describe successive computations, they all refer to single-round executions and the respective results.
The bit-wise manipulation does not happen over time, instead, we aim at examining closely related circuit states.

We begin with Lemma~\ref{lem:pivotal} which applies to a single round of computation.
It states that feeding a circuit $C$ with a pivotal sequence $x$ of states results in a pivotal sequence of possible successor states~$y$.
Hence, if $C$ is guaranteed to output different results for $x^{(0)}$ and~$x^{(\ell)}$, some intermediate element of $y$ must contain a metastable output bit, i.e., there is an execution in which an output register of $C$ becomes metastable.
We argue about successor states rather than just the output because we inductively apply Lemma~\ref{lem:pivotal} in Corollary~\ref{cor:pivotal}.
A sample circuit with pivotal sequences is depicted in Figure~\ref{fig:pivotal}.

Let $x$ be a pivotal sequence of non-output register states, i.e., over $\BM^{m+n}$, and suppose a pivotal register changes from stable to \meta from $x^{(i)}$ to $x^{(i+1)}$.
By Observation~\ref{obs:read}, we may construct executions with $x^{(i)} \in \Read^C(x^{(i)})$ and $x^{(i+1)} \in \Read^C(x^{(i+1)})$.
The key insight is that due to $x^{(i)} \in \ResM(x^{(i+1)})$, we have $f^G(x^{(i)}) \in \ResM(f^G(x^{(i+1)}))$ by Lemma~\ref{lem:dag}.
Hence, $\Write^C(x^{(i)}) \subseteq \Write^C(x^{(i+1)})$ as $\Write^C(x^{(i)}) = \ResM(f^G(x^{(i)})$ and $\Write^C(x^{(i+1)}) = \ResM(f^G(x^{(i+1)})$ by Lemma~\ref{lem:1-round-writes}.
This destabilizes the bits that are affected by the destabilized input bit in the successor states;
we leave all other bits unchanged.
Leveraging this, we obtain a pivotal sequence of successor states, changing the affected output bits from stable to \meta one by one, each result of a one-round execution of~$C$.
A reversed version of this argument applies when a non-output register changes from \meta to stable.

\begin{figure}
\hfill\subfigure[Circuit]{
	\begin{tikzpicture}[scale=1,transform shape,cktbaselength=0.5pt]
		\draw (0,0.5) node[or2,draw,rotate=90] (or) {};
		\draw (1.5,0) node[and2,draw,rotate=90] (and) {};

		\draw (or.b) -- ++(-0.6,0) node[anchor=east, left] (i1) {$I_1$};
		\draw (or.a) -- (i1.east |- or.a) node[anchor=east, left] (i2) {$I_2$};
		\draw (and.a) -- (i2.east |- and.a) node[anchor=east, left] (l1in) {$L_1$};
		\draw (or.z) -- ++(2,0) node[anchor=west, right] (l1out) {$L_1$};

		\coordinate (split) at ($ (or.z) + (0.4,0) $);
		\draw (split) node[circle, fill=black, inner sep=0, minimum size=3pt] {}
			-- (split |- and.b)
			-- (and.b);
		\draw (and.z) -- (l1out.west |- and.z)
			node[anchor=west, right] (o1) {$O_1$};
	\end{tikzpicture}
\label{fig:pivotal-circuit}
}\hfill\subfigure[Pivotal sequences of register states]{
	\begin{tabular}{rccc|cc|ccccl}
		& \multicolumn{3}{c}{$x^{(i)}$} & \multicolumn{2}{|c|}{$f^G(x^{(i)})$} & \multicolumn{4}{c}{$y^{(j)}$} & \\
		        & $I_1$ &   $I_2$ &   $L_1$ &   $L_1$ &   $O_1$ & $I_1$ &   $I_2$ &   $L_1$ &   $O_1$ & \\
		\hline
		$x^{(0)}$ & $0$ &     $0$ &     $0$ &     $0$ &     $0$ &   $0$ &     $0$ &     $0$ &     $0$ & $y^{(0)}$ \\
		$x^{(1)}$ & $0$ &     $0$ & $\meta$ &     $0$ &     $0$ &   $0$ &     $0$ &     $0$ &     $0$ & \\
		$x^{(2)}$ & $0$ &     $0$ &     $1$ &     $0$ &     $0$ &   $0$ &     $0$ &     $0$ &     $0$ & \\
		$x^{(3)}$ & $0$ & $\meta$ &     $1$ & $\meta$ & $\meta$ &   $0$ & $\meta$ &     $0$ &     $0$ & $y^{(1)}$ \\
		          & $0$ & $\meta$ &     $1$ & $\meta$ & $\meta$ &   $0$ & $\meta$ & $\meta$ &     $0$ & $y^{(2)}$ \\
		          & $0$ & $\meta$ &     $1$ & $\meta$ & $\meta$ &   $0$ & $\meta$ & $\meta$ & $\meta$ & $y^{(3)}$ \\
		          & $0$ & $\meta$ &     $1$ & $\meta$ & $\meta$ &   $0$ & $\meta$ &     $1$ & $\meta$ & $y^{(4)}$ \\
		          & $0$ & $\meta$ &     $1$ & $\meta$ & $\meta$ &   $0$ & $\meta$ &     $1$ &     $1$ & $y^{(5)}$ \\
		$x^{(4)}$ & $0$ &     $1$ &     $1$ &     $1$ &     $1$ &   $0$ &     $1$ &     $1$ &     $1$ & $y^{(6)}$ \\
	\end{tabular}
\label{fig:pivotal-table}
}\hfill{}
\caption{%
	A circuit with input ($I_1$~and~$I_2$), local ($L_1$), and output ($O_1$) registers~\subref{fig:pivotal-circuit}, and a pivotal sequence of non-output register states $x^{(0)}, \dots, x^{(4)}$ with the resulting pivotal sequence of successor states $y^{(0)}, \dots, y^{(6)}$~\subref{fig:pivotal-table}.
	Each change in $x$ is reflected in a re-evaluation of the combinational logic $f^G(x^{(i)})$, which may affect several registers of the successor state.
	In order to be pivotal, the output sequence $y$ accounts for the changes bit by bit.%
}
\label{fig:pivotal}
\end{figure}

\begin{lemma}\label{lem:pivotal}
	Let $C$ be a circuit, and
	\begin{equation}
		\left(x^{(i)}\right)_{i \in [\ell+1]}, \quad x^{(i)} \in \BM^{m+k+n},
	\end{equation}
	a pivotal sequence of states of~$C$.
	Then there is a pivotal sequence
	\begin{equation}
		\left(y^{(j)}\right)_{j \in [\ell'+1]}, \quad y^{(j)} \in \BM^{m+k+n},
	\end{equation}
	where each $y^{(j)}$ is a successor state of some $x^{(i)}$, satisfying that $y^{(0)}$ and $y^{(\ell')}$ are successor states of~$x^{(0)}$ and~$x^{(\ell)}$, respectively.
\end{lemma}

\begin{proof}
	See Figure~\ref{fig:pivotal} for an illustration of our arguments.
	Starting from~$x^{(0)}$, we inductively proceed to~$x^{(\ell)}$, extending the sequence $y$ by a suitable subsequence for each step from $x^{(i)}$ to~$x^{(i+1)}$, $i \in [\ell]$.
	We maintain the invariant that the state $y^{(j)}$ corresponding to $x^{(i)}$ fulfills
	\begin{equation}
		\Loc\left( y^{(j)} \right) \circ \Out\left( y^{(j)} \right)
			= f^G \left( \In\left( x^{(i)} \right) \circ \Loc\left( x^{(i)}\right) \right).
	\end{equation}
	Let $\iota = \In(x^{(0)})$ be the state of the input registers.
	By Lemma~\ref{lem:1-round-writes}, $f^G(\iota \circ \Loc(x^{(0)})) \in \Write^C(x^{(0)})$.
	Define $y^{(0)} = \iota' \circ f^G(\iota \circ \Loc(x^{(0)}) \in S_1^C(x^{(0)})$, where $\iota'$ is the uniquely determined state of the input registers after reading~$\iota$.
	By construction, $x^{(0)}$ and $y^{(0)}$ fulfill the invariant.

	We perform the step from $x^{(i)}$ to $x^{(i+1)}$, $i \in [\ell]$.
	Let $P_i$ be the pivotal register from $x^{(i)}$ to $x^{(i+1)}$;
	if $P_i$ is an output register, $y$ does not change, so assume that $P_i$ is an input or local register.
	From the previous step (or the definition of~$y^{(0)}$) we have an execution resulting in state $y^{(j)}$ for some index~$j$, such that $x^{(i)}$ is the result of the read phase.
	For the next step, we keep the result of the read phase for all registers except $P_i$ fixed.
	Regarding all registers that do not depend on~$P_i$, i.e., may attain the same states regardless of what is read from~$P_i$, we rule that they attain the same states as in~$y^{(j)}$, the state associated with~$x^{(i)}$.

	Suppose first that $x_{P_i}^{(i)} = b \neq \meta$ and $x_{P_i}^{(i+1)} = \meta$ (e.g.\ the step from $x^{(2)}$ to $x^{(3)}$ in Figure~\ref{fig:pivotal}).
	Consider the set of non-input registers $\mathcal{R}$ that depend on~$P_i$, i.e., $\mathcal{R} := \{ R \mid f^G(x^{(i)})_R \neq f^G(x^{(i+1)})_R \}$ ($\mathcal{R} = \{L_1, O_1\}$ in our example).
	Since $x^{(i)} \in \ResM(x^{(i+1)})$, by Lemma~\ref{lem:dag} $f^G(x^{(i)}) \in \ResM(f^G(x^{(i+1)}))$.
	Hence, $f^G(x^{(i+1)})_R = \meta \neq f^G(x^{(i)})_R$ for all $R \in \mathcal{R}$.

	If $P_i$ is an input register, we first extend $y$ by one item that only changes $y_{P_i}$ to~\meta, increase $j$ by one if that is the case (e.g.\ the step from $y^{(0)}$ to $y^{(1)}$ in Figure~\ref{fig:pivotal}).
	Then we extend $y$ by $y^{(j+1)}, \dots, y^{(j+|\mathcal{R}|)}$ such that in each step, for one $R \in \mathcal{R}$, we change $y_R$ from $b_R \neq \meta$ to~\meta;
	this is feasible by Corollary~\ref{cor:bitwise-closed}, as the product structure of $C_1$ implies that we can flip any written bit without affecting the others (e.g.\ steps $y^{(2)}$ and $y^{(3)}$ in our example).
	By construction, in state~$y^{(j+|\mathcal{R}|)}$ the state of the non-input registers is $f^G(\iota \circ x^{(i+1)})$, i.e., our invariant is satisfied.

	To cover the case that $x_{P_i}^{(i)} = \meta$ and $x_{P_i}^{(i+1)} = b \neq \meta$, observe that we can apply the same reasoning by reversing the order of the constructed attached subsequence.
	As $y$ is pivotal by construction, this completes the proof.
\end{proof}

Given a pivotal sequence of inputs, there are executions producing a pivotal sequence of attainable successor states.
Using these states for another round, Lemma~\ref{lem:pivotal} can be applied inductively.

\begin{corollary}\label{cor:pivotal}
	Let $C$ be a circuit, $x_0$~its initialization, and
	\begin{equation}
		\left( \iota^{(i)} \right)_{i \in [\ell+1]}, \quad \iota^{(i)} \in \BM^m,
	\end{equation}
	be a pivotal sequence of inputs of~$C$.
	Then there is a pivotal sequence of states
	\begin{equation}
		\left( y^{(j)} \right)_{j \in [\ell'+1]}, \quad y^{(j)} \in \BM^{m+k+n},
	\end{equation}
	that $C$ can attain after $r \in \N$ rounds satisfying $y^{(0)} \in S^C_r(\iota^{(0)} \circ x_0)$ and $y^{(\ell')} \in S^C_r(\iota^{(\ell)} \circ x_0)$.
\end{corollary}

\begin{proof}
	Inductive application of Lemma~\ref{lem:pivotal} to $C$ and states $\left( \iota^{(i)} \circ x_0 \right)_{i \in [\ell+1]}$.
\end{proof}

We wrap up our results in a compact theorem.
It states that a circuit which has to output different results for different inputs can produce metastable outputs.
\begin{theorem}\label{thm:pivotal}
	Let $C$ be a circuit with $C_r(\iota) \cap C_r(\iota') = \emptyset$ for some $\iota, \iota' \in \BM^m$.
	Then $C$ has an $r$-round execution in which an output register becomes metastable.
\end{theorem}

\begin{proof}
	Apply Corollary~\ref{cor:pivotal} to a pivotal sequence from $\iota$ to $\iota'$ and~$C$, yielding a pivotal sequence $y$ of states that $C$ can attain after $r$-round executions.
	Since $C_r(\iota) \ni \Out(y^{(0)}) \neq \Out(y^{(\ell')}) \in C_r(\iota')$, some $\Out(y^{(j)})$ contains a metastable bit.
\end{proof}

Marino proved that no digital circuit, synchronous or not, can reliably
\begin{enumerate*}
\item\label{enum:marino-avoid}
	compute a non-constant function and guarantee non-metastable output,

\item\label{enum:marino-detect}
	detect whether a register is metastable, or

\item\label{enum:marino-resolve}
	resolve metastability of the input while faithfully propagating stable input~\cite{m-gtmo-81}.
\end{enumerate*}
Theorem~\ref{thm:pivotal} captures~\ref{enum:marino-avoid}, and Corollaries~\ref{cor:no-meta-detection} and~\ref{cor:no-meta-resolve} settle~\ref{enum:marino-detect} and~\ref{enum:marino-resolve}, respectively.
The key is to observe that a circuit detecting or resolving metastability is non-constant, and hence, by Theorem~\ref{thm:pivotal}, can become metastable\dash---defeating the purpose of detecting or resolving metastability in the first place.

\begin{corollary}\label{cor:no-meta-detection}
	There exists no circuit that implements $f\colon \BM \to \Pow(\BM)$ with
	\begin{equation}
		f(x) = \begin{cases}
				\{1\} & \text{if $x = \meta$, and} \\
				\{0\} & \text{otherwise.}
			\end{cases}
	\end{equation}
\end{corollary}

\begin{proof}
	Assume such a circuit $C$ exists and implements $f$ in $r$ rounds.
	$C_r(0) \cap C_r(\meta) = \emptyset$, so applying Theorem~\ref{thm:pivotal} to $\iota = 0$ and $\iota' = \meta$ yields that $C$ has an $r$-round execution with metastable output, contradicting the assumption.
\end{proof}

\begin{corollary}\label{cor:no-meta-resolve}
	There exists no circuit that implements $f\colon \BM \to \mathcal{P}(\BM)$ with
	\begin{equation}
		f(x) = \begin{cases}
				\{0,1\} & \text{if $x = \meta$, and} \\
				\{x\}   & \text{otherwise.}
			\end{cases}
	\end{equation}
\end{corollary}

\begin{proof}
	As in Corollary~\ref{cor:no-meta-detection} with $\iota = 0$ and $\iota' = 1$.
\end{proof}

In summary, our circuit model (Section~\ref{sec:model}) is consistent with physical models of metastability, yet admits the computation of non-trivial functions (Section~\ref{sec:mux}) that are crucial in constructing complex metastability-containing circuits~\cite{blm-nomcsn-17,fklp-mametdc-17,lm-emcgc2s-16}.
This gives rise to further questions:
\begin{enumerate*}
\item
	Is there a fundamental difference between simple and masking registers?
\item
	Which functions can be implemented?
\end{enumerate*}
We study these questions in Sections~\ref{sec:hierarchy} and~\ref{sec:simple}, respectively.

\section{Computational Hierarchy}
\label{sec:hierarchy}

In this section, we determine the impact of the number of rounds $r \in \N$ and the available register types on the computational power of a circuit.
Recall that $\Fun_S^r$ denotes the functions implementable using $r$ rounds and simple registers only, and $\Fun_M^r$ those implementable using $r$ rounds and arbitrary registers.
The main results are the following.
\begin{enumerate}
\item
	Even in the presence of metastability, circuits restricted to simple registers can be unrolled (Section~\ref{sec:hierarchy-simple}): $\Fun_S^r = \Fun_S^{r+1}$.

\item
	With masking registers, however, more functions become implementable with each additional round (Section~\ref{sec:hierarchy-arbitrary}): $\Fun_M^r \subsetneq \Fun_M^{r+1}$.
\end{enumerate}
Together with Corollary~\ref{cor:1-round-equivalence}, we obtain the following hierarchy:
\begin{equation}
	\cdots = \Fun_S^2 = \Fun_S^1 = \Fun_M^1 \subsetneq \Fun_M^2 \subsetneq \cdots.
\end{equation}
We believe this to make a strong case for further pursuing masking registers in research regarding metastability-containing circuits.

\subsection{Simple Registers}
\label{sec:hierarchy-simple}

It is folklore that binary-valued synchronous circuits can be unrolled such that the output after $r \in \N$ clock cycles of the original circuit is equal to the output after a single clock cycle of the unrolled circuit.
Theorem~\ref{thm:unrolling} states that this result also holds in presence of potentially metastable simple registers.
Note that\dash---defying intuition\dash---masking registers do not permit this, see Theorem~\ref{thm:diff}.

\begin{figure}
	\begin{center}
		\def\svgwidth{.7\linewidth}
		{\small 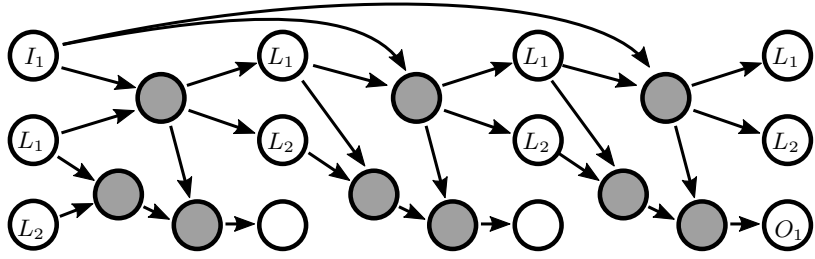}
	\end{center}
	\caption{%
		Unrolling three rounds of the circuit in Figure~\ref{fig:dag} with three gates (gray), and four registers (white).
		Local registers become fan-out buffers, and early output is ignored.}
	\label{fig:unrolling}
\end{figure}

\begin{theorem}\label{thm:unrolling}
	Given a circuit $C$ with only simple registers such that $r \in \N$ rounds of $C$ implement~$f$, one can construct a circuit $C'$ such that one round of $C'$ implements~$f$.
\end{theorem}

% This proof sketch is up to date with the full proof below as of 2016/03/18
%\begin{proof}[Sketch]
%	Arrange $r$ copies of the combinational logic of $C$ as in Figure~\ref{fig:unrolling} such that
%	\begin{enumerate*}
%	\item
%		input registers feed all copies of gates they feed in~$C$,
%	\item
%		local registers become fan-out buffers (gates forwarding their input), and
%	\item
%		output registers are copied as well, but only the $r$-th copy is relevant.
%	\end{enumerate*}
%	We have $C_1' = C_r$ because simple registers merely maintain and propagate metastability in the worst case.
%\end{proof}

\begin{proof}
	We construct a circuit $C'$ with $C'_1(\iota) = C_2(\iota)$; the claim for $r > 2$ then follows by induction.

	Given~$C$, we construct $C'$ as follows, compare Figure~\ref{fig:unrolling}.
	Let $G$ be the combinational logic \ac{DAG} of~$C$, make two copies $G_1 = (V_1, A_1)$ and $G_2 = (V_2, A_2)$ of~$G$, and let $G' = (V_1 \cup V_2, A_1 \cup A_2)$ be the combinational logic \ac{DAG} of~$C'$, up to the following modifications.
	Every input register $I$ of $C$ corresponds to input nodes $v_1^I \in V_1$ and $v_2^I \in V_2$.
	Contract $\{v_1^I, v_2^I\}$ to a single input node in~$G'$ (compare $I_1$ in Figure~\ref{fig:unrolling}), and associate it with a new input register in~$C'$;
	repeat this for all input registers.
	In order to ignore ``early'' output, replace each output node in $G_1$ corresponding to an output register in $C$ with a gate that has one input and whose output is ignored (like the first two copies of $O_1$ in Figure~\ref{fig:unrolling}).
	The remaining input and output nodes are associated with local registers.
	Each local register $L$ of $C$ corresponds to exactly one output node $v_1^L \in V_1$ and one input node $v_2^L \in V_2$.
	Contract $\{v_1^L, v_2^L\}$ to a fan-out buffer gate that simply forwards its input in $G'$ (the center copies of $L_1$ and $L_2$ in Figure~\ref{fig:unrolling}).
	Associate the $k$ remaining input nodes of $G_1$ and output nodes of $G_2$ with local registers.
	Observe that $G'$ has $n$ input, $m$~output, and $k$ local registers.
	Define the initial state of $C'$ as that of~$C$.

	To check that one round of $C'$ is equivalent to two rounds of~$C$, let $\iota \in \BM^m$ be an input, $s_0$ the initial state of both $C$ and $C'$ w.r.t.\ input~$\iota$, and $x_0 = \In(s_0) \circ \Loc(s_0)$.
	First recall that by definition (Figure~\ref{fig:register-states-simple}), simple registers never change their state when read.
	Hence by construction of~$G'$, we have $\Eval^{C'}(s_0) = \{ f^{G'}(x_0) \} = \{ f^G(\iota \circ \Loc(f^G(x_0))) \}$.

	In~$C$, we have $\Write^C(s_0) = \ResM(f^G(x_0))$ by Lemma~\ref{lem:1-round-writes}.
	Thus, in the second round of $C$, for any $s_1 \in S_1^C$ we have that $\Read^C(s_1) = \{ \iota \circ \Loc(s_1) \} \subseteq \ResM(\iota \circ \Loc(f^G(x_0)))$, and by Lemma~\ref{lem:dag} $\Eval^C(s_1) \subseteq \ResM(f^G(\iota \circ \Loc(f^G(x_0))))$.
	This means that the second evaluation phase of $C$ yields a stabilization of the first evaluation phase of~$C'$, i.e., $S_2^C \subseteq S_1^{C'}$, because the write phase allows for arbitrary stabilization.

	On the other hand, the unstabilized $\iota \circ \Loc(f^G(x_0)) \in \Eval^C(s_1)$, so $S_1^{C'} \subseteq S_2^C$.
	Together, we have $S_2^C = S_1^{C'}$ and $C'_1(\iota) = C_2(\iota)$ follows.
\end{proof}

Naturally, the unrolled circuit can be significantly larger than the original one.
However, the point is that adding rounds does not affect the computational power of circuits with simple registers only.
\begin{corollary}\label{cor:funs}
	For all $r \in \N$, $\Fun_S^r = \Fun_S^1 =: \Fun_S$.
\end{corollary}

\subsection{Arbitrary Registers}
\label{sec:hierarchy-arbitrary}

For simple registers, additional rounds make no difference in terms of computability\dash---the corresponding hierarchy collapses into $\Fun_S$.
In the following, we demonstrate that this is not the case in the presence of masking registers:
$\Fun_M^r \subsetneq \Fun_M^{r+1}$ for all $r \in \N$.
We demonstrate this using a metastability-containing fan-out buffer specified by Equation~\eqref{eq:masking-fanout}.
It creates $r$ copies of its input bit, at most one of which is permitted to become metastable:
\begin{equation}\label{eq:masking-fanout}
	f(x) = \begin{cases}
		\{ x^r \}
			& \text{if $x \neq \meta$,} \\
		\bigcup_{i \in [r]} \ResM(0^i \meta 1^{r-i-1})
			& \text{otherwise.}
	\end{cases}
\end{equation}

\begin{figure}
	\begin{center}
		\def\svgwidth{.55\linewidth}
		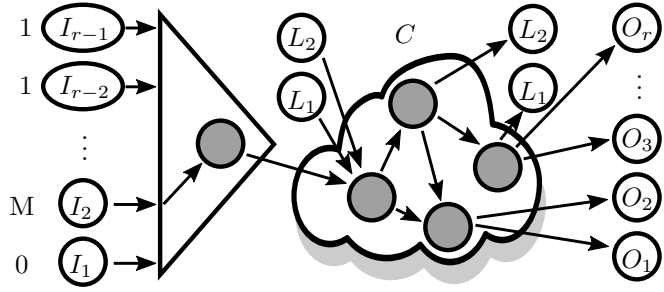
	\end{center}
	\caption{Simulating a masking register with a selector.}
	\label{fig:diff}
\end{figure}

\begin{theorem}\label{thm:diff}
	$\Fun_M^r \subsetneq \Fun_M^{r+1}$ for all $r \in \N$.
\end{theorem}

% This sketch is new as of 2016-03-18, use it as proof or outside of the proof
% environment to replace the lengthy proof below in the conference version.
%\begin{proof}[Sketch]
%	Pick $2 \leq r \in \N$ and consider $f$ from Equation~\eqref{eq:masking-fanout}.
%	To see that $f \in \Fun_M^r$, have a circuit $C$ store the input in a mask-$0$ register, and read one copy of it in each of $r$ rounds.
%	If $x \neq \meta$, $r$ rounds of $C$ generate $r$ stable copies of~$x$.
%	Otherwise $x = \meta$ and the $r$ outputs are specified by $r$ state transitions of the mask-$0$ register starting in state~\meta, i.e., behave exactly as specified in~\eqref{eq:masking-fanout}.
%
%	As for $f \notin \Fun_M^{r-1}$, assume $r-1$ rounds of $C$ implement~\eqref{eq:masking-fanout} and observe that the input register $R$ can only be read $r-1$ times.
%	Since $C$ produces $r$ outputs, two of these outputs have to depend on the same read of~$R$.
%	If that read operation returns~\meta, which is possible even for masking registers, both outputs can become metastable, violating the specification~\eqref{eq:masking-fanout}.
%\end{proof}

\begin{proof}
	Fix $2 \leq r \in \N$ and consider $f$ from~\eqref{eq:masking-fanout}.
	We first show $f \in \Fun_M^r$, and then that $f \notin \Fun_M^{r-1}$.

	$f$~is implemented by $r$ rounds of the circuit $C$ which uses a mask-$0$ input register $R_{r-1}$ and a chain of local registers $R_{r-2}, \dots, R_0$.
	In each round, the value read from register~$R_{i+1}$, $i \in [r-1]$, is copied to~$R_i$, and output register~$O_i$, $i \in [r]$, gets the value read from~$R_i$.
	Observe that the specification of a mask-$0$ register is such that, given an initial state, $r$~reads (and possibly stabilization in the write phase) may return exactly the sequences specified in~\eqref{eq:masking-fanout}.
	Since $C$ faithfully copies these values, it follows that $f = C_r \in \Fun_M^r$.

	We claim that for $r \geq 2$, $f \notin \Fun_M^{r-1}$.
	Assume for contradiction that there is a circuit $C$ such that $C_{r-1} \subseteq f$.
	We derive a contradiction by simulating the behavior of $C$ in a circuit $C'$ with $r - 1$ simple input registers, which may initially hold any possible sequence of values read from the input register of $C$ in $r - 1$ rounds.

	To specify this circuit, we first observe that the following subcircuits are straightforward to implement:
	\begin{description}
	\item [$r$-round counters]
		take no input and have $r$ outputs, such that the $i$-th output is $1$ in round $1 \leq i \leq r$ and $0$ else.
		This is implemented by a linear chain of local registers~$R_i$, $i \in [r]$ (i.e., $R_i$ is copied to $R_{i+1}$ for $i \in [r-1]$), where $R_0$ is initialized to $1$ and all others to~$0$, output~$O_{i+1}$, $i \in [r]$, is fed the \gxor of $R_i$ and~$R_{i+1}$, and $R_{r-1}$ is copied to~$O_r$.

	\item [$r$-round selectors]
		take $r$ inputs~$x_i$, $i \in [r]$, and have one output~$O$, such that the state of~$O$ in round $1 \leq i \leq r$ is in $\ResM(x_{i-1})$ (i.e., holds a copy of~$x_{i-1}$).
		This is achieved by using an $r$-round counter and feeding the \gand of $x_i$ and $c_i$ (the $i$-th counter output) into an $r$-ary \gor-gate whose output is written into~$O$.
	\end{description}

	By Corollary~\ref{cor:register-equivalence}, we may assume w.l.o.g.\ that all non-input registers of $C$ are simple.
	If the input register is also simple, $\meta^r \in C_{r-1}(\meta) \notin f(\meta)$ by Theorem~\ref{thm:unrolling} and Lemma~\ref{lem:1-round-writes}.

	Consider the case that the input register is a mask-$0$ register and compare Figure~\ref{fig:diff}.
	Define $C'$ as a copy of~$C$, except that $r - 1$ simple input registers serve as input to an $(r-1)$-round selector.
	This compound represents the only input register $R$ of~$C$:
	Every gate or output node driven by $R$ in $C$ is instead wired to the selector's output in~$C'$.

	A surjective mapping of executions of $C'$ with inputs restricted to $\{ 0^i \meta 1^{r-i-1} \mid i \in [r-1] \}$, i.e., all possible reads from $R$ in state~\meta, to executions of $C$ is defined as follows.
	We interpret the selector's output in round $r$ as the value read from $R$ in round $r$ and ``copy'' the remaining execution of $C'$ (without inputs and the selector) to obtain a complete execution of~$C$.
	Due to our restriction on the inputs, the result always is a feasible execution of $C$ with input~\meta.

	By Theorem~\ref{thm:unrolling}, we may w.l.o.g.\ assume that a single round of $C'$ implements~$f$.
	Consider the sequence of $C'$-inputs from $0^{r-1}$ to $1^{r-1}$ in which we flip the bits one by one from right to left, from $0$ to~$1$.
	By the pigeon hole principle, there must be some $1 \leq \bar{r} \leq r-1$ so that two output bits of $C'$ change compared to $\bar{r} - 1$.
	Since, when fixing the other input bits, two outputs $\ell \neq \ell'$ depend on the $\bar{r}$-th input bit and $C'$ only uses simple registers, we have by Lemma~\ref{lem:1-round-writes} that $\meta\meta \in \Write^{C'}(0^{r-1-\bar{r}} \meta 1^{\bar{r}-1})_{\ell,\ell'}$.
	Hence, $\ell$ and $\ell'$ can become metastable in the same execution of~$C'$.
	We map this execution to an execution of~$C$, in which the corresponding output registers attain the same state (i.e., two of them are~\meta) after $r - 1$ rounds.
	This covers the case that the input register is a mask-$0$ register;
	a mask-$1$ register is handled analogously.

	We arrive at the contradiction that $C_{r-1} \not\subseteq f$, implying that $f \notin \Fun_M^{r-1}$.
	Overall, $\Fun_M^{r-1} \neq \Fun_M^r$.
	As $r \geq 2$ was arbitrary and, by Observation~\ref{obs:subseteq}, $\Fun_M^{r-1} \subseteq \Fun_M^r$, this concludes the proof.
\end{proof}

\section{The Power of Simple Registers}
\label{sec:simple}

The design of metastability-containing circuits requires a quick and easy check which metastability-containing components are implementable, and which are not.
In this section, we present such a test for circuits without masking registers.

First, we present sufficient and necessary conditions for a function to be implementable with simple registers only (Section~\ref{sec:simple-natural}).
Using this classification, we demonstrate how to take an arbitrary Boolean function $f\colon \B^m \to \B^n$ and extend it to the most restrictive specification $[f]_\meta\colon \BM^m \to \Pow(\BM^n)$, the metastable closure of~$f$, that is implementable. This is an easy process\dash---one simply applies Definition~\ref{def:closure} to $f$ (Section~\ref{sec:simple-closure}).

The way to make use of this is to start with a function $f$ required as component, ``lift'' it to~$[f]_\meta$, and check whether $[f]_\meta$ is restrictive enough for the application at hand.
If it is, one can work on an efficient implementation of~$[f]_\meta$, otherwise a new strategy, possibly involving masking registers, must be devised;
in either case, no time is wasted searching for a circuit that does not exist.
Sections~\ref{sec:simple-closure-possible}--\ref{sec:simple-closure-impossible} summarize our findings.

Since we discuss functions implementable with simple registers only, recall that the corresponding circuits can be unrolled by Theorem~\ref{thm:unrolling}, i.e., it suffices to understand~$C_1$, a single round of a (possibly unrolled) circuit.

\subsection{Natural Subfunctions}
\label{sec:simple-natural}

From Corollary~\ref{cor:bitwise-closed} and Observation~\ref{obs:specific}, we know that~$C_1$, the set of possible circuit outputs after a single round, has three properties:
\begin{enumerate*}
\item
	its output can be specified bit-wise,

\item
	each output bit is either~$0$, $1$, or completely unspecified, and

\item
	stabilizing a partially metastable input restricts the set of possible outputs.
\end{enumerate*}
Hence $C_1$\dash---and by Corollary~\ref{cor:funs} all circuits using only simple registers\dash---can be represented in terms of bit-wise \ac{KV} diagrams with values ``$0$,\,$1$,\,$\BM$'' instead of ``$0$,\,$1$,\,D'' (D~for ``don't care'').
We call such functions \emph{natural} and show below that $f \in \Fun_S$ if and only if $f$ has a natural subfunction.

\begin{definition}[Natural and Subfunctions]\label{def:natural}
	The function $f\colon \BM^m \to \Pow(\BM^n)$ is \emph{natural} if and only if it is \emph{bit-wise,} \emph{closed,} and \emph{specific:}
	\begin{description}
	\item [Bit-wise]
		The components $f_1, \dots, f_n$ of $f$ are independent:
		\begin{equation}\label{eq:bit-wise}
			f(x) = f_1(x) \times \dots \times f_n(x).
		\end{equation}

	\item [Closed]
		Each component of~$f$ is specified as either~$0$, as~$1$, or completely unspecified:
		\begin{equation}\label{eq:closed}
			\forall x \in \BM^m\colon\quad
				f(x) \in \{ \{0\}, \{1\}, \BM \}^n.
		\end{equation}

	\item [Specific]
		When stabilizing a partially metastable input, the output of $f$ remains at least as restricted:
		\begin{equation}\label{eq:specific}
			\forall x \in \BM^m\colon \quad
				x' \in \Res(x) \Rightarrow f(x') \subseteq f(x).
		\end{equation}
	\end{description}
	For functions $f,g\colon \BM^m \to \Pow(\BM^n)$, $g$~is a \emph{subfunction of~$f$} (we write $g \subseteq f$), if and only if $g(x) \subseteq f(x)$ for all $x \in \BM^m$.
\end{definition}

Suppose we ask whether a function $f$ is implementable with simple registers only, i.e., if $f \in \Fun_S$.
Since any (unrolled) circuit $C$ implementing $f$ must have $C_1 \subseteq f$, Corollary~\ref{cor:bitwise-closed} and Observation~\ref{obs:specific} state a necessary condition for $f \in \Fun_S$:
$f$~must have a natural subfunction.
Theorem~\ref{thm:simple-comp} establishes that this condition is sufficient, too.

\begin{theorem}\label{thm:simple-comp}
	Let $g\colon \BM^m \to \Pow(\BM^n)$ be a function.
	Then $g \in \Fun_S$ if and only if $g$ has a natural subfunction.
\end{theorem}

\begin{proof}
	For the only-if-direction, suppose that $C$ is a circuit with only simple registers such that $C_1 \subseteq g$;
	by Theorem~\ref{thm:unrolling}, such a circuit exists.
	$C_1$ is bit-wise and closed by Corollary~\ref{cor:bitwise-closed}, and specific by Observation~\ref{obs:specific}.
	Hence, choosing $f := C_1$ yields a natural subfunction of~$g$.

	We proceed with the if-direction.
	Let $f \subseteq g$ be a natural subfunction of~$g$, and construct a circuit $C$ that implements~$f$.
	As $f$ is bit-wise, we may w.l.o.g.\ assume that $n = 1$.
	If $f(\cdot) = \{ 0 \}$ or $f(\cdot) = \BM$, let $C$ be the circuit whose output register is driven by a \gzero-gate;
	if $f(\cdot) = \{ 1 \}$, use a \gone-gate.
	Otherwise, we construct $C$ as follows.
	Consider $f_{\B}\colon \B^m \to \{ \{ 0 \}, \{ 1 \} \}$ given by
	\begin{equation}
		f_{\B}(x) = \begin{cases}
			\{ 0 \} & \text{if $f(x) = \{ 0 \}$ or $f(x) = \BM$, and} \\
			\{ 1 \} & \text{if $f(x) = \{ 1 \}$.}
		\end{cases}
	\end{equation}
	We call a partial variable assignment $A$ that implies $f_\B(x) = \{ 1 \}$ for all $x$ obeying $A$ an \emph{implicant of~$f_\B$;}
	if the number of variables fixed by $A$ is minimal w.r.t.\ $A$ being an implicant, we call $A$ a \emph{prime implicant of~$f_\B$.}
	Construct $C$ from \gand-gates, one for each prime implicant of~$f_{\B}$, with inputs connected to the respective, possibly negated, input registers present in the prime implicant.
	All \gand-gate outputs are fed into a single \gor-gate driving the circuit's only output register.

	By construction, $C_1(x) = f_{\B}(x) \subseteq f(x)$ for all $x \in \B^m$.
	To see $C_1 \subseteq f$, consider $x \in \BM^m \setminus \B^m$ and make a case distinction.
	\begin{enumerate}
	\item
		If $f(x) = \BM$, then trivially $C_1(x) \subseteq f(x)$.

	\item
		If $f(x) = \{ 0 \}$, we have for all $x' \in \Res(x)$ that $f(x') = f_{\B}(x') = \{ 0 \}$ by~\eqref{eq:specific}.
		Thus, for each such~$x'$, all \gand-gate outputs are~$0$.
		Furthermore, under input~$x$ and for each \gand-gate, there must be at least one input that is stable~$0$:
		Otherwise, there would be some $x' \in \Res(x)$ making one \gand-gate output~$1$, resulting in $f_{\B}(x') = \{1\}$.
		By our definition of gate behavior, this entails that all \gand-gates output~$0$ for all $x' \in \ResM(x)$ as well, and hence $C_1(x) = \{ 0 \} = f(x)$.

	\item
		If $f(x) = \{ 1 \}$, all $x' \in \Res(x)$ have $f(x') = f_{\B}(x') = \{1\}$ by~\eqref{eq:specific}.
		Thus, $f_\B$~outputs $\{ 1 \}$ independently from the metastable bits in~$x$, and there is a prime implicant of $f_\B$ which relies only on stable bits in~$x$.
		By construction, some \gand-gate in $C$ implements that prime implicant.
		This \gand-gate receives only stable inputs from~$x$, and hence outputs a stable~$1$.
		The \gor-gate receives that $1$ as input and, by definition of gate behavior, outputs stable~$1$.
		Hence, $C_1(x) = \{ 1 \} = f(x)$.
	\end{enumerate}
	As $f$ is closed, this case distinction is exhaustive.
	The claim follows as one round of $C$ implements~$f$.
\end{proof}

Theorem~\ref{thm:simple-comp} is useful for checking if a circuit implementing some function \emph{exists;} its proof is constructive.
However, we obtain no non-trivial bound on the \emph{size} of the respective circuit\dash---covering all prime implicants can be costly.
While efficient metastability-containing implementations exist~\cite{blm-nomcsn-17,lm-emcgc2s-16}, it is an open question
\begin{enumerate*}
\item
	which functions can be implemented efficiently \emph{in general,} and

\item
	what the overhead for metastability-containment w.r.t.\ an implementation oblivious to metastability is.
\end{enumerate*}

\subsection{Metastable Closure}
\label{sec:simple-closure}

We propose a generic method of identifying and creating functions implementable with simple registers.
Consider a classical Boolean function $f\colon \B^m \to \B^n$ defined for stable in- and outputs only.
Lift the definition of $f$ to $[f]_\meta$ dealing with (partly) metastable inputs analogously to gate behavior in Section~\ref{sec:model-gates}:
Whenever all metastable input bits together can influence the output, specify the output as ``anything in~$\BM$.''
We call $[f]_\meta$ the \emph{metastable closure of~$f$,} and argue below that $[f]_\meta \in \Fun_S$.
For $f\colon \BM^m \to \Pow(\BM^n)$, i.e., for more flexible specifications, $[f]_\meta$~is defined analogously.

\begin{definition}[Metastable Closure]\label{def:closure}
	For a function $f\colon \BM^m \to \Pow(\BM^n)$, we define its \emph{metastable closure} $[f]_\meta \colon \BM^m \to \Pow(\BM^n)$ component-wise for $i \in [n]$ by
	\begin{equation}\label{eq:closure}
		[f]_\meta(x)_i := \begin{cases}
			\{ 0 \} & \text{if $\forall x' \in \ResM(x)\colon\; f(x')_i = \{ 0 \}$,} \\
			\{ 1 \} & \text{if $\forall x' \in \ResM(x)\colon\; f(x')_i = \{ 1 \}$,} \\
			\BM     & \text{otherwise.}
		\end{cases}
	\end{equation}
	We generalize~\eqref{eq:closure} to Boolean functions.
	For $f\colon \B^m \to \B^n$, we define $[f]_\meta\colon \BM^m \to \Pow(\BM^n)$ as
	\begin{equation}\label{eq:closure-bool}
		[f]_\meta(x)_i := \begin{cases}
			\{ 0 \} & \text{if $\forall x' \in \Res(x)\colon\; f(x')_i = 0$,} \\
			\{ 1 \} & \text{if $\forall x' \in \Res(x)\colon\; f(x')_i = 1$,} \\
			\BM     & \text{otherwise.}
		\end{cases}
	\end{equation}
\end{definition}

By construction, $[f]_\meta$ is bit-wise, closed, specific, and hence natural.
\begin{observation}\label{obs:closure-simple}
	$[f]_\meta \in \Fun_S$ for all $f\colon \B^m \to \B^n$ and for all $f\colon \BM^m \to \Pow(\BM^n)$.
\end{observation}

\subsubsection{Showing what is Possible}
\label{sec:simple-closure-possible}

An immediate consequence of Observation~\ref{obs:closure-simple} for the construction of circuits is that, given an arbitrary Boolean function $f\colon \B^m \to \B^n$, there is a circuit without masking registers that implements~$[f]_\meta$.

For $f\colon \B^m \to \B^n$, Theorem~\ref{thm:simple-comp} shows that $[f]_\meta$ is the minimum extension of $f$ implementable with simple registers: by \eqref{eq:specific} any natural extension $g$ of $f$ must satisfy
\begin{equation}
	\forall x \in \BM^m, \forall i \in [n]\colon \quad
		\bigcup_{x' \in \Res(x)} f(x')_i \subseteq g(x)_i,
\end{equation}
and thus $\exists x',x'' \in \Res(x)\colon f(x')_i \neq f(x'')_i \Rightarrow g(x)_i = \BM$ by~\eqref{eq:closed}.

\subsubsection{Showing what is Impossible}
\label{sec:simple-closure-impossible}

In order to show that a function is not implementable with simple registers only, it suffices to show that it violates the preconditions of Theorem~\ref{thm:simple-comp}, i.e., that it has no natural subfunction.
\begin{example}
	Consider $f\colon \BM^2 \to \Pow(\BM^2)$ with
	\begin{equation}\label{eq:mm-example}
		f(x) := \ResM(x) \setminus \{ \meta\meta \}.
	\end{equation}
	This function specifies to copy a $2$-bit input, allowing metastability to resolve to anything except~$\meta\meta$.
	No circuit without masking registers implements~$f$: $f \notin \Fun_S$.
\end{example}

The recipe to prove such a claim is:
\begin{enumerate}
\item
	For contradiction, assume $f \in \Fun_S$, i.e., that $f$ has some natural subfunction $g \subseteq f$ by Theorem~\ref{thm:simple-comp}.

\item
	By specification of~$f$, the individual output bits of $g$ can become metastable for input~$\meta\meta$.

\item
	Since $g$ is bit-wise, it follows that $\meta\meta \in g(\meta\meta)$.

\item
	This contradicts the assumption that $g \subseteq f$.
\end{enumerate}

\begin{proof}
	Assume for contradiction $f \in \Fun_S$, i.e., that $f$ has a natural subfunction $g \subseteq f$ by Theorem~\ref{thm:simple-comp}.

	As specified in~\eqref{eq:mm-example}, $f(00)_1 = \{0\}$ and $f(11)_1 = \{1\}$.
	Since $g \subseteq f$ and $g(x) \neq \emptyset$ because $g$ is closed, we have $g(00)_1 = \{0\}$ and $g(11)_1 = \{1\}$.
	The fact that $g$ is specific implies that $g(00)_1 \cup g(11)_1 \subseteq g(\meta\meta)_1$, i.e., $\{0,1\} \in g(\meta\meta)_1$.
	This in turn means that $g(\meta\meta)_1 = \BM$, because $g$ is closed.
	Furthermore, we know that $g$ is bit-wise, so $g = g_1 \times g_2$ with $g_1(00) = \{0\}$, $g_1(11) = \{1\}$, and $g_1(\meta\meta) = \BM$.
	Analogously, $g_2(00) = \{0\}$, $g_2(11) = \{1\}$, and $g_2(\meta\meta) = \BM$.

	Since $g$ is bit-wise, $g(\meta\meta) = g_1(\meta\meta) \times g_2(\meta\meta) \ni \meta\meta$, but $\meta\meta \notin f(\meta\meta)$, contradicting the assumption $g \subseteq f$.
	As we did not make any restrictions regarding~$g$, this holds for all natural subfunctions of~$f$.
	It follows that $f \notin \Fun_S$.
\end{proof}

\section{Components for Clock Synchronization}
\label{sec:arithmetic}

This section demonstrates the power of our techniques:
We establish that a variety of metastability-containing components are a reality.
Due to the machinery established in the previous sections, this is possible with simple checks (usually using Observation~\ref{obs:closure-simple}).
The list of components is by no means complete, but already allows implementing a highly non-trivial application.

We are the first to demonstrate the physical implementability of the fault-tolerant clock synchronization algorithm by Lundelius Welch and Lynch~\cite{ll-ftacs-88} with deterministic correctness guarantee, despite the unavoidable presence of metastable upsets.
The algorithm of Lundelius Welch and Lynch is widely applied, e.g., applied in the \ac{TTP}~\cite{kb-tta-03} and in FlexRay~\cite{bbb+-fcp-03}.
While the software--hardware based implementations of \ac{TTP} and FlexRay achieve a precision in the order of microseconds, higher operating frequencies ultimately require a pure hardware implementation.
Recently, an implementation of the algorithm of Lundelius Welch and Lynch based on \iac{FPGA} has been presented by Kinali et al.~\cite{hkl-ftcshp-16}.
All known implementations, however, synchronize potentially metastable inputs \emph{before} computations\dash---a technique that becomes less reliable with increasing operating frequencies, since less time is available for metastability resolution.

Moreover, classical bounds for the \ac{MTBF} for metastable upsets assume a uniform distribution of input transitions;
this is not guaranteed to be the case in clock synchronization, since the goal is to align clock ticks.
Either way, synchronizers do not deterministically guarantee stabilization, and errors are bound to happen eventually when $n$ clocks take $n(n-1)$ samples at, e.g., $1$\,GHz.
The combination of ever-increasing operating frequencies and the inevitability~\cite{m-gtmo-81} of metastable upsets when measuring relative timing deviations leads us to a fundamental question:
Does the unavoidable presence of metastable upsets pose a principal limit on the operating frequency?
We show that this is not the case.

\subsection{Algorithm and Required Components}
\label{sec:arithmetic-application}

Our core strategy is the \emph{separation of concerns} outlined in Section~\ref{sec:introduction} and Figure~\ref{fig:clocksync}.
The key is that the digital part of the circuit can become metastable, but that metastability is \emph{contained} and ultimately translated into \emph{bounded} fluctuations in the analog world, not contradicting Marino.

We propose an implementation for $n$ clock-synchronization nodes with at most $f < n/3$ faulty nodes, in which each node does the following.

\subsubsection{Step~1: Analog to Digital}
\label{sec:arithmetic-application-tdc}

First, we step from the analog into the digital world:
Delays between $n-1$ remote pulses and the local pulse are measured with \acp{TDC}.
The measurement can be realized such that at most one of the output bits, accounting for the difference between $x$ and $x+1$ ticks, becomes metastable;
we say such numbers have \emph{precision-$1$} and formally define them in Section~\ref{sec:arithmetic-encoding}.

\begin{figure}
	\begin{center}
		{\small \begin{tikzpicture}[scale=1,transform shape,cktbaselength=0.5pt]
	\tikzstyle{gate} = [draw,anchor=center]
	\tikzstyle{ann} = [below,font=\small\tt]
	\tikzstyle{latch} = [draw,anchor=center,rectangle,text width=5mm,text height=8mm]

	% start
	\draw (-0.4,0) node[label={[xshift=-0.3cm,yshift=-0.3cm] remote start}] (i1) {};

	\draw (1,0) node[buffer,gate,rotate=90] (b1) {};
	\draw (2,0) node[buffer,gate,rotate=90] (b2) {};
	\draw (3,0) node[buffer,gate,rotate=90] (b3) {};
	\draw (4,0) node[buffer,gate,rotate=90] (b4) {};

	\draw (b1.a) -- ++(-0.5,0);
	\draw (b1.z) -- (b2.a);
	\draw (b2.z) -- (b3.a);
	\draw (b3.z) -- (b4.a);
	\draw (b4.z) -- ++(0.5,0); %node[xshift=4mm] {$out$};

	% latch 1
	\draw (1.0,-1) node[latch] (l1) {};
	\draw (l1) node[xshift=-2mm,yshift=2mm] {\footnotesize $D$};
	\draw (l1) node[xshift=2mm,yshift=1.6mm] {\footnotesize $Q$};
	\draw (l1) node[xshift=-2mm,yshift=-2.1mm] {\footnotesize $\bar E$};
	\draw (0.5,0) -- ++(0,-0.8) -- ++(0.12,0);
	\draw[fill] (0.5,0) circle (1.5pt);

	% latch 2
	\draw (2.0,-1) node[latch] (l1) {};
	\draw (l1) node[xshift=-2mm,yshift=2mm] {\footnotesize $D$};
	\draw (l1) node[xshift=2mm,yshift=1.6mm] {\footnotesize $Q$};
	\draw (l1) node[xshift=-2mm,yshift=-2.1mm] {\footnotesize $\bar E$};
	\draw (1.5,0) -- ++(0,-0.8) -- ++(0.12,0);
	\draw[fill] (1.5,0) circle (1.5pt);

	% latch 3
	\draw (3.0,-1) node[latch] (l1) {};
	\draw (l1) node[xshift=-2mm,yshift=2mm] {\footnotesize $D$};
	\draw (l1) node[xshift=2mm,yshift=1.6mm] {\footnotesize $Q$};
	\draw (l1) node[xshift=-2mm,yshift=-2.1mm] {\footnotesize $\bar E$};
	\draw (2.5,0) -- ++(0,-0.8) -- ++(0.12,0);
	\draw[fill] (2.5,0) circle (1.5pt);

	% latch 4
	\draw (4.0,-1) node[latch] (l1) {};
	\draw (l1) node[xshift=-2mm,yshift=2mm] {\footnotesize $D$};
	\draw (l1) node[xshift=2mm,yshift=1.6mm] {\footnotesize $Q$};
	\draw (l1) node[xshift=-2mm,yshift=-2.1mm] {\footnotesize $\bar E$};
	\draw (3.5,0) -- ++(0,-0.8) -- ++(0.12,0);
	\draw[fill] (3.5,0) circle (1.5pt);

	% latch 5
	\draw (5.0,-1) node[latch] (l1) {};
	\draw (l1) node[xshift=-2mm,yshift=2mm] {\footnotesize $D$};
	\draw (l1) node[xshift=2mm,yshift=1.6mm] {\footnotesize $Q$};
	\draw (l1) node[xshift=-2mm,yshift=-2.1mm] {\footnotesize $\bar E$};
	\draw (4.5,0) -- ++(0,-0.8) -- ++(0.12,0);
	\draw[fill] (4.5,0) circle (1.5pt);

	% stop
	\draw (-0.4,-2) node[label={[xshift=-0.2cm,yshift=-0.3cm] local stop}] (i2) {};
	\path (i2) -- ++(0.5,0) node (ii) {};
	\draw (ii) -- ++(4.5,0);

	% taps
	\draw (0.5,-2) -- ++(0,0.8) -- ++(0.12,0);
	\draw[fill] (0.5,-2) circle (1.5pt);

	\draw (1.5,-2) -- ++(0,0.8) -- ++(0.12,0);
	\draw[fill] (1.5,-2) circle (1.5pt);

	\draw (2.5,-2) -- ++(0,0.8) -- ++(0.12,0);
	\draw[fill] (2.5,-2) circle (1.5pt);

	\draw (3.5,-2) -- ++(0,0.8) -- ++(0.12,0);
	\draw[fill] (3.5,-2) circle (1.5pt);

	\draw (4.5,-2) -- ++(0,0.8) -- ++(0.12,0);
	\draw[fill] (4.5,-2) circle (1.5pt);
\end{tikzpicture}} % match caption size
	\end{center}
	\caption{%
		Tapped delay line \acs{TDC}.
		It is read as either $1^k 0^{n-k}$ or $1^k \meta 0^{n-k-1}$, i.e., produces at most one metastable bit and hence has precision-$1$.%
	}
	\label{fig:tdc}
\end{figure}
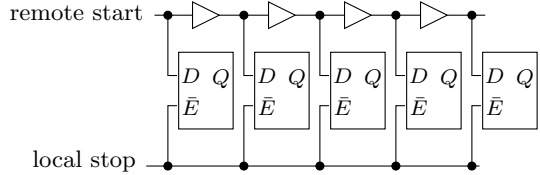

\acp{TDC} can be implemented using tapped delay lines or Vernier delay line \acp{TDC}~\cite{glnhc-stcibr-94,rk-scdldsti-93,ra-bitdcdttc-10}; see Figure~\ref{fig:tdc}:
A line of delay elements is tapped in between each two consecutive elements, driving the data input port of initially enabled latches initialized to~$0$.
The rising transition of the remote clock signal fed into the delay line input then passes through the line, and sequentially sets the latches to~$1$;
the rising transition of the local clock signal is used to disable all latches at once.
After that, the delay line's latches contain the time difference as unary \ac{TC}.
Choosing the propagation delays between the latches larger than their setup/hold times, we ensure that at most one bit is metastable, i.e., their status is of the form $1^*0^*$ or $1^* \meta 0^*$.
The output is hence a precision-$1$ \acs{TC}-encoded time difference.

A traditional implementation would use synchronizers on the \ac{TDC} outputs.
This delays the computation and encourages stabilization, but does not enforce it.
However, clock synchronization cannot afford to wait.
Furthermore, we prefer guaranteed correctness over a probabilistic statement:
Four nodes, each sampling at $1$\,GHz, sample $1.2 \cdot 10^{10}$ incoming clock pulses per second;
synchronizers cannot provide sufficiently small error probabilities when allocating $1$\,ns or less for metastability resolution~\cite{bgpdk-ds-10}.
Hence, we consider the use of metastability-containing arithmetic instead of synchronizers mandatory.

\subsubsection{Step~2: Encoding}

We translate the time differences into \ac{BRGC}, making storage and subsequent components much more efficient.
The results are \acs{BRGC}-encoded time differences with at most one metastable bit of precision-$1$.

In this step, metastability-containing \ac{TC} to \ac{BRGC} conversion is needed and we discuss it in Section~\ref{sec:arithmetic-converter}.
A more efficient way is to use a metastability-containing \ac{TDC} which directly produces \ac{BRGC} of precision-$1$;
such a component is presented in~\cite{fklp-mametdc-17}.

\subsubsection{Step~3: Sorting Network}

A sorting network selects the $(f+1)$-th and $(n-f)$-th largest remote-to-local clock differences (tolerating $f$ faults requires to discard the smallest and largest $f$ values).

This requires $2$-sort building blocks that pick the minimum and maximum of two precision-$1$ \acs{BRGC}-encoded inputs preserving precision-$1$.
We discuss this in Section~\ref{sec:arithmetic-sorting};
efficient implementations are given in~\cite{blm-nomcsn-17,lm-emcgc2s-16} and improved in~\cite{fk-emcm-17}.

\subsubsection{Step~4: Decoding and Digital to Analog}

The \acs{BRGC}-encoded $(f+1)$-th and $(n-f)$-th largest remote-to-local clock differences are translated back to \acs{TC}-encoded numbers.
As discussed in Section~\ref{sec:arithmetic-converter2}, this can be done preserving precision-$1$, i.e., such that the results are of the form $1^* 0^*$ or $1^* \meta 0^*$.

Finally, we step back into the analog world, again without losing precision:
The two values are used to control the local clock frequency via \iac{DCO}.
However, the \ac{DCO} design must be chosen with care.
Designs that switch between inverter chains of different length to modify the frequency of a ring oscillator cannot be used, as metastable switches may occur exactly when a pulse passes.
Instead, we use a ring oscillator whose frequency is controlled by analog effects such as changes in inverter load or bias current, see e.g.~\cite{dvrc-fdcoir-07,on-dcpllsoca-04,zk-dcolpc-08}.
While the at most two metastable control bits may dynamically change the load of two inverters, this has a limited effect on the overall frequency change and does not lead to glitches within the ring oscillator.

Carefully note that this gives a \emph{guaranteed end-to-end uncertainty of a single bit} through all digital computations.

\subsection{Encoding and Precision}
\label{sec:arithmetic-encoding}

An appropriate encoding is key to designing metastability-containing arithmetic components.
If, for example, a control bit $u$ indicating whether to increase $x = 7$ by $1$ is metastable, and $x$ is encoded in binary, the result must be a metastable superposition of $00111$ and~$01000$, i.e., anything in $\Res(0\meta\meta\meta\meta)$ and thus an encoding of any number $x' \in [16]$\dash---even after resolving metastability!
The original uncertainty between $7$ and $8$ is massively amplified;
a good encoding should \emph{contain} the uncertainty imposed by $u = \meta$.

Formally, a \emph{code} is an injective function $\gamma\colon [n] \to \B^k$ mapping a natural number $x \in [n]$ to its encoded representation.
For $y = \gamma(x)$, we define $\gamma^{-1}(y) := x$, and for sets~$X$, $\gamma(X) := \{ \gamma(x) \mid x \in X \}$ and $\gamma^{-1}(X) := \{x \mid \gamma(x) \in X \}$.
In this work, we consider two encodings for input and output: \ac{TC} and \ac{BRGC}.
For the $4$-bit (unary) \ac{TC} we use $\un\colon [5] \to \B^4$ with $\un(1) = 0001$ and $\un^{-1}(0111) = 3$; $\un^{-1}(0101)$ does not exist.
\ac{BRGC}, compare Figure~\ref{fig:gray-code}, is represented by~$\rg(x)$, and is much more efficient, using only $\lceil \log_2 n \rceil$ bits.
In fact, $\rg\colon [2^k] \to \B^k$ is bijective.

We choose $\un$ and $\rg$ due to the property that in both encodings, for $x \in [k-1]$, $\gamma(x)$ and $\gamma(x+1)$ differ in a single bit only.
This renders them suitable for metastability-containing operations.
We revisit the above example with the metastable control bit $u$ indicating whether to increase $x = 7$ by~$1$.
In \ac{BRGC}, $7$~is encoded as $00100$ and $8$ as~$01100$, so their metastable superposition resolves to $\Res(0\meta100)$, i.e., only to $7$ or~$8$.
Since the original uncertainty was whether or not to increase $x = 7$ by~$1$, the uncertainty is perfectly contained instead of amplified as above.
We formalize the notion of the amount of uncertainty in a partially metastable code word:
$x \in \BM^k$ has \emph{precision-$p$ (w.r.t.\ the code~$\gamma$)} if
\begin{equation}\label{eq:precision}
	\max\left\{ y - \bar{y} \mid y, \bar{y} \in \gamma^{-1}(\Res(x)) \right\} \leq p,
\end{equation}
i.e., if the largest possible difference between resolutions of $x$ is bounded by~$p$.
The precision of $x$ w.r.t.\ $\gamma$ is undefined if some $y \in \Res(x)$ is no code word, which is not the case in our application.

Note that the arithmetic components presented below make heavy use of \ac{BRGC}.
This makes them more involved, but they are exponentially more efficient than their \ac{TC} counterparts in terms of memory and avoid the amplification of uncertainties incurred by standard binary encoding.
As a matter of fact, recently proposed efficient implementations for metastability-containing sorting networks~\cite{blm-nomcsn-17,lm-emcgc2s-16} and metastability-containing \acp{TDC}~\cite{fklp-mametdc-17} use \ac{BRGC}.

\subsection{Digital Components}
\label{sec:arithmetic-components}

In the following, we show that all metastability-containing components required for the clock synchronization algorithm outlined in Section~\ref{sec:arithmetic-application} exist.
As motivated above, the components have to maintain meaningful outputs in face of limited metastability;
more precisely, we deal with precision-$1$ inputs due to the nature of \acp{TDC} (see Section~\ref{sec:arithmetic-application-tdc}).
Note that this section greatly benefits from the machinery established in previous sections\dash---in particular from Observation~\ref{obs:closure-simple} which immediately shows which components exist.

\subsubsection{Thermometer to Binary Reflected Gray Code}
\label{sec:arithmetic-converter}

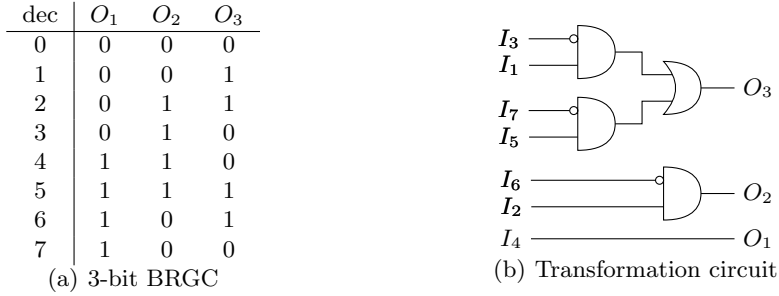
\begin{figure}
	{\small % match caption size
	\hfill\subfigure[$3$-bit \acs{BRGC}]{
		\begin{tabular}{c|ccc}
			dec & $O_1$ & $O_2$ & $O_3$\\
			\hline
			$0$ & $0$ & $0$ & $0$ \\
			$1$ & $0$ & $0$ & $1$ \\
			$2$ & $0$ & $1$ & $1$ \\
			$3$ & $0$ & $1$ & $0$ \\
			$4$ & $1$ & $1$ & $0$ \\
			$5$ & $1$ & $1$ & $1$ \\
			$6$ & $1$ & $0$ & $1$ \\
			$7$ & $1$ & $0$ & $0$ \\
		\end{tabular}
		\label{fig:gray-code}
	}\hfill
	\subfigure[Transformation circuit]{
	\raisebox{-1.55cm}{
		\begin{tikzpicture}[scale=1,transform shape,cktbaselength=0.5pt]
			% o1
			\draw (4,0) node (o1) {$O_1$};
			\draw (0.75,0) node (i4) {$I_4$};
			\path[draw] (o1) -- (i4);

			% o2
			\draw (4,0.6) node (o2) {$O_2$};
			\draw (3,0.6) node[and2ni,draw,rotate=90] (a1) {};
			\draw (a1.z) -- (o2);
			\path (a1.a) -- ++(-2,0) node (i2) {$I_2$};
			\path (a1.b) -- ++(-1.9,0) node (i6) {$I_6$};
			\draw (i2) node {$I_2$};
			\draw (i6) node {$I_6$};
			\draw (a1.a) -- (i2);
			\draw (a1.b) -- (i6);

			% o3
			\draw (4,2) node (o3) {$O_3$};
			\draw (3,2) node[or2,draw,rotate=90] (or1) {};
			\draw (or1.z) -- (o3);
			\path (or1.b) -- ++(-1,0.3) node[and2ni,draw,rotate=90] (a2) {};
			\path (or1.a) --  ++(-1,-0.3) node[and2ni,draw,rotate=90] (a3) {};
			\draw (or1.b) -- ++(-0.4,0) -- ++(0,0.3) -- (a2.z);
			\draw (or1.a) -- ++(-0.4,0) -- ++(0,-0.3) -- (a3.z);

			\path (a2.a) -- ++(-0.9,0) node (i1) {$I_1$};
			\path (a2.b) -- ++(-0.8,0) node (i3) {$I_3$};
			\path (a3.a) -- ++(-0.9,0) node (i5) {$I_5$};
			\path (a3.b) -- ++(-0.8,0) node (i7) {$I_7$};
			\draw (i1) node {$I_1$};
			\draw (i3) node {$I_3$};
			\draw (i5) node {$I_5$};
			\draw (i7) node {$I_7$};
			\draw (a2.a) -- (i1);
			\draw (a2.b) -- (i3);
			\draw (a3.a) -- (i5);
			\draw (a3.b) -- (i7);
		\end{tikzpicture}
		}
		\label{fig:gray-circuit}
	}\hfill}
	\caption{Efficient \acs{TC}-to-\acs{BRGC} conversion.}
\label{fig:gray}
\end{figure}

At the hand of the example circuit in Figure~\ref{fig:gray}, we show how precision-$1$ \ac{TC}-encoded data can be efficiently translated into precision-$1$ \acs{BRGC}-encoded data. % (a general result for arbitrary Gray codes is shown in~\cite{tdc}).
Figure~\ref{fig:gray-circuit} depicts the circuit that translates a $7$-bit \ac{TC} into a $3$-bit \ac{BRGC};
note that gate count and depth are optimal for a fan-in of~$2$.
The circuit can be easily generalized to $n$-bit inputs, having a gate depth of $\lfloor \log_2 n \rfloor$.
While such translation circuits are well-known, it is important to check that the given circuit fulfills the required property of preserving precision-$1$:
This holds as each input bit influences exactly one output bit, and, due to the nature of \ac{BRGC}, this bit makes exactly the difference between $\rg(x)$ and $\rg(x+1)$ given a \acs{TC}-encoded input of $1^x \meta 0^{7-x-1}$.

\subsubsection{Sorting Networks}
\label{sec:arithmetic-sorting}

It is well-known that sorting networks can be efficiently composed from $2$-sort building blocks~\cite{aks-asn-83,b-sna-68}, which map $(x,y)$ to $(\min\{x,y\}, \max\{x,y\})$.
We show that $\max$ (and analogously $\min$) of two precision-$1$ $k$-bit \ac{BRGC} numbers is implementable without masking registers, such that each output has precision-$1$.
Observe that this is straightforward for \acs{TC}-encoded inputs with bit-wise \gand and \gor for $\min$ and $\max$, respectively.
We show, however, that this is possible for \ac{BRGC} inputs as well;
efficient implementations of the proposed $2$-sort building blocks are presented in~\cite{blm-nomcsn-17,lm-emcgc2s-16}.

\begin{lemma}\label{lem:max}
	Define $\max_\text{\acs{BRGC}}\colon \B^k \times \B^k \to \B^k$ as
	\begin{equation}
		{\max}_\text{\acs{BRGC}}(x,y) := \rg\left( \max\left\{ \rg^{-1}(x), \rg^{-1}(y) \right\} \right).
	\end{equation}
	Then $[\max_\text{\acs{BRGC}}]_\meta \in \Fun_S$ and it determines precision-$1$ output from precision-$1$ inputs $x$ and~$y$.
\end{lemma}

\begin{proof}
	Since $x$ and $y$ have precision-$1$, $\rg^{-1}(\Res(x)) \subseteq \{a, a+1\}$ for some $a \in [2^k - 1]$ (analogously for $y$ w.r.t.\ some $b \in [2^k - 1]$).
	W.l.o.g.\ assume $a \geq b$, i.e., for all possible resolutions of $x$ and~$y$, the circuit must output $\rg(a)$ or $\rg(a+1)$.
	By Definition~\ref{def:closure} and the fact that $\rg(a)$ and $\rg(a+1)$ differ in a single bit only, $[\max_\text{BRGC}]_\meta(x,y)$ has at most one metastable bit and precision-$1$.
\end{proof}

An analogous argument holds for
\begin{equation}
	{\min}_\text{\acs{BRGC}}(x,y) := \rg\left( \min\left\{ \rg^{-1}(x), \rg^{-1}(y) \right\} \right).
\end{equation}

\subsubsection{Binary Reflected Gray to Thermometer Code}
\label{sec:arithmetic-converter2}

A \acs{BRGC}-encoded number of precision-$1$ has at most one metastable bit:
For any up-count from (an encoding of) $x \in [2^k-1]$ to $x + 1$, a single bit changes, which thus can become metastable if it has precision-$1$.
It is possible to preserve this guarantee when converting to \ac{TC}.
\begin{lemma}
	Define $\rgtoun\colon \B^k \to \B^{(2^k-1)}$ as
	\begin{equation}
		\rgtoun(x) := \un\left( \rg^{-1}(x) \right).
	\end{equation}
	Then $[\rgtoun]_\meta \in \Fun_S$ converts its parameter to \ac{TC}, preserving precision-$1$.
\end{lemma}

\begin{proof}
	If $x$ has precision-$1$, then $\rg^{-1}(\Res(x)) \subseteq \{a,a+1\}$ for some $a \in [2^k - 1]$.
	Hence, $\un(a)$ and $\un(a+1)$ differ in a single bit, proving the claim.
\end{proof}

\section{Conclusion}
\label{sec:conclusion}

No digital circuit can reliably avoid, detect, or resolve metastable upsets~\cite{m-gtmo-81}.
So far, the only known counter strategy has been to use synchronizers\dash---trading time for an increased probability of resolving metastability.
We propose a fundamentally different method:
It is possible to design efficient digital circuits that tolerate a certain degree of metastability in the input.
This technique features critical advantages:
\begin{enumerate}
\item
	Where synchronizers decrease the odds of failure, our techniques provide deterministic guarantees.
	A synchronizer may or may not stabilize in the allotted time frame.
	Our model, on the other hand, guarantees to return one of a specific set of known values\dash---like the metastable closure, but this depends on the application\dash---without relying on probabilities.

\item
	Our approach avoids synchronization delay and, in principle, allows higher operating frequencies.
	If the required functions can be implemented in a metastability-containing way, there is no need to use a synchronizer, i.e., to wait a fixed amount of clock cycles before starting the computation.

\item
	Even if metastability needs to be resolved eventually, one can still save time by allowing for stabilization \emph{during} the metastability-containing computations.
\end{enumerate}
In light of these properties, we expect our techniques to prove useful for a variety of applications, especially in time- and mission-critical scenarios.

As a consequence of our techniques, we are the first to establish the implementability of the fault-tolerant clock synchronization algorithm by Lundelius Welch and Lynch~\cite{ll-ftacs-88} with a deterministic correctness guarantee, despite the unavoidable presence of metastable upsets.

Furthermore, we fully classify the functions computable with circuits restricted to standard registers.
Finally, we show that circuits with masking registers become computationally more powerful with each round, resulting in a non-trivial hierarchy of computable functions.

%###
\paragraph*{Future Work}
%###

In this work, we focus on computability under metastable inputs.
There are many open questions regarding circuit complexity in our model of computation.
It is of interest to reduce the gate complexity and latency of circuits, as well as to determine the complexity overhead of metastability-containment in general.
In particular, the overhead of implementing the metastable closure $[f]_\meta$ as compared to an implementation of $f$ that is oblivious to metastability\dash---and if there has to be an overhead at all\dash---is an open question, in general as well as for particular functions~$f$.
Recently, promising results have been obtained for sorting networks~\cite{blm-nomcsn-17,lm-emcgc2s-16}, \iac{TDC} that directly produces precision-$1$ \ac{BRGC}~\cite{fklp-mametdc-17}, and network-on-chip routers~\cite{tfl-mtc-17}.

Masking registers are computationally strictly more powerful than simple registers (Theorem~\ref{thm:diff}).
An open question is which metastability-containing circuits benefit from masking registers, and to find further examples separating $\Fun_M^r$ from $\Fun_M^{r+1}$.

Our model does not capture clock gating, i.e., non-input registers are overwritten in every clock cycle.
Hence, storing intermediate results in masking registers is pointless:
Taking advantage of at most one read from an input masking-register becoming metastable does not apply to results of intermediate computations.
It is an open problem whether this makes a difference in terms of computability.

\section*{Acknowledgments}

The authors would like to thank Attila Kinali, Ulrich Schmid, and Andreas Steininger for many fruitful discussions.
Matthias F{\"u}gger was affiliated with Max Planck Institute for Informatics during the research regarding this paper.

\bibliographystyle{abbrv}
\bibliography{bibliography}

\end{document}